\documentclass[12pt,psamsfonts,reqno]{amsart}
\usepackage{amsmath,amsthm,amsfonts,amssymb,chngcntr,relsize,caption}
\usepackage[title]{appendix}
\usepackage{eucal}
\usepackage{bbm}
\usepackage{placeins}
\usepackage{graphicx}
\usepackage{xcolor}
\usepackage{longtable}
\usepackage[all,knot]{xy}
\usepackage[pdftex,plainpages=false,hypertexnames=false,pdfpagelabels]{hyperref}
\definecolor{medium-blue}{rgb}{0,0,.8}
\hypersetup{
   colorlinks, linkcolor={purple},
   citecolor={medium-blue}, urlcolor={medium-blue}
}
\xyoption{arc}

\usepackage{fullpage}

\newcommand{\Z}{\mathbb{Z}}
\newcommand{\R}{\mathbb{R}}
\newcommand{\C}{\mathbb{C}}
\newcommand{\cC}{\mathcal{C}}
\newcommand{\bbX}{\mathbb{X}}
\newcommand{\bbH}{\mathbb{H}}
\newcommand{\cM}{\mathcal{M}}
\newcommand{\cP}{\mathcal{P}}
\newcommand{\cD}{\mathcal{D}}
\newcommand{\cE}{\mathcal{E}}

\newcommand{\GL}{\operatorname{GL}}

\newcommand{\PSL}{\operatorname{PSL}}

\newcommand{\ch}{\operatorname{ch}}
\newcommand{\Span}{\operatorname{span}}
\newcommand{\Rep}{\operatorname{Rep}}

\newcommand{\rank}{\operatorname{rank}}
\newcommand{\rk}{\operatorname{rank}}
\newcommand{\Fib}{\operatorname{Fib}}
\newcommand{\Semion}{\operatorname{Semion}}
\newcommand{\LeeYang}{\operatorname{Yang-Lee}}

\newcommand{\DHaag}{\operatorname{DHaag}}
\newcommand{\oVec}{\operatorname{Vec}}

\newcommand{\abs}[1]{\left| #1 \right|}

\newenvironment{psmallmatrix}
  {\left(\begin{smallmatrix}}
  {\end{smallmatrix}\right)}

\newcommand{\sfr}[2]{\leavevmode\kern-.1em
  \raise.5ex\hbox{\the\scriptfont0 #1}\kern-.1em
  /\kern-.15em\lower.25ex\hbox{\the\scriptfont0 #2}}

\numberwithin{equation}{section}

\newtheorem*{maintheorem}{Main Theorem}

\newtheorem{theorem}{Theorem}[section]

\newtheorem{lemma}[theorem]{Lemma}
\newtheorem{lem}[theorem]{Lemma}

\theoremstyle{definition}

\theoremstyle{remark}

\newenvironment{tablenofloat}{\captionsetup{type=table}}{}

\begin{document}
\title[]{Classification of extremal vertex operator algebras with two simple modules}

\author{J. Connor Grady}
\address{Grady: Department of Mathematics, University of Illinois at Urbana-Champaign,
Urbana, IL 61801 USA}
\email{jcgrady2@illinois.edu}

\author{Ching Hung Lam}
\address{Lam: Institute of Mathematics, Academia Sinica, Taipei 10617, Taiwan}
\email{chlam@math.sinica.edu.tw}

\author{James E. Tener}
\address{Tener: Mathematical Sciences Institute, The Australian National University, Canberra, ACT 2600, Australia}
\email{james.tener@anu.edu.au}

\author{Hiroshi Yamauchi}
\address{Yamauchi: Department of Mathematics, Tokyo Woman's Christian University, 2-6-1 Zempukuji, Suginami-ku, Tokyo 167-8585, Japan}
\email{yamauchi@lab.twcu.ac.jp}

\thanks{
The authors would like to thank Kyoto RIMS for its hospitality which facilitated this collaboration.
The third author would also like to thank Terry Gannon for helpful discussions related to this article.
The second author was partially supported by research grants AS-IA-107-M02 from  Academia Sinica, Taiwan and 104-2115-M-001-004-MY3 from Ministry of Sciences and Technology of Taiwan.
The third author was supported in part by an AMS-Simons Travel Grant and by Australian Research Council Discovery Project DP200100067.
The fourth author was partially supported by JSPS KAKENHI grant No.19K03409.}


\begin{abstract}
In recent work, Wang and the third author defined a class of `extremal' vertex operator algebras (VOAs), consisting of those with at least two simple modules and conformal dimensions as large as possible for the central charge.
In this article we show that there are exactly 15 character vectors of extremal VOAs with two simple modules.
All but one of the 15 character vectors is realized by a previously known VOA.
The last character vector is realized by a new VOA with central charge 33.
\end{abstract}

\maketitle

\section{Introduction}
In 1988, a foundational article of Mathur, Mukhi, and Sen \cite{MathurMukhiSen88} pioneered an approach for the classification of rational chiral conformal field theories (CFTs).
They studied the graded dimension functions, or characters, of CFTs by showing that they satisfied what are now called modular linear differential equations (MLDEs).
The classification of CFTs by their characters, via MLDEs and other methods, has continued over the last two decades, and in recent years there has been significant activity surrounding classification of characters of chiral CFTs with two characters.

We study a mathematical version of this classification problem.
We take vertex operator algebras (VOAs) as a mathematical model for chiral CFTs, and for a sufficiently nice (`strongly rational') vertex operator algebra $V$ of central charge $c$, we consider its representation category $\cC = \Rep(V)$, which is a modular tensor category \cite{HuangModularity}.
One can recover the equivalence class of the central charge $c$ mod $8$ from $\cC$, and motivated by this we define an \emph{admissible genus} to be a pair $(\cC, c)$ consisting of a modular tensor category and a number $c$ in the appropriate class mod $8$ (c.f. \cite{Hoehn03}).
It is natural to approach the problem of classification of VOAs and their characters by restricting to genera where both $\cC$ and $c$ are sufficiently small in an appropriate sense.

We will take the rank of an MTC (i.e. the number of simple objects) as our measure of its size.
The smallest MTC is the trivial one, $\oVec$, and a genus $(\oVec, c)$ is admissible when $c \equiv 0$ mod $8$.
There are a total of three VOAs in the genera $(\oVec, 8)$ and $(\oVec, 16)$, arising from even unimodular lattices of rank $8$ and $16$.
The problem becomes interesting at $(\oVec, 24)$, where the classification of characters of VOAs can be read off from Schellekens' famous list \cite{Schellekens93} (see also \cite{vanEkerenMoellerScheithauer20}).
The classification of VOAs in $(\oVec, 24)$ is almost complete, except that the uniqueness of VOA(s) with the same character as the Moonshine VOA has not been established.
At $(\oVec, 32)$ the explicit classification problem of characters is already intractable, even just for the simplest examples coming from even unimodular lattices.

This difficulty propagates to the study of genera $(\cC, c)$ with $c$ large, as for any fixed $V \in (\cC, c)$ one obtains from every $W \in (\oVec, 32)$ a new VOA $V \otimes W \in (\cC, c+32)$.
Thus if one wishes to consider classification problems for higher central charge and $\rank(\cC) > 1$, it is necessary to restrict to a class of VOAs which excludes VOAs like $V \otimes W$.
There is a natural notion of `primeness' that one could consider in this context, but we will consider something slightly different.

The twist of a simple module $M \in \Rep(V)$ is given by $\theta_M = e^{2 \pi i h}$, where $h$ is the lowest conformal dimension of states in $M$.
Thus for any (hypothetical) VOA $V \in (\cC, c)$, we can recover the conformal dimensions of simple objects, mod 1.
Moreover, there is an a priori bound \cite{MathurMukhiSen88, Mason07} on the conformal dimensions:%
\footnote{%
We are assuming that the characters of $V$ are linearly independent, which will always hold in the situation under consideration in this article, when $V$ has exactly two simple modules}%
\begin{equation}\label{eqn: ell}
\ell :=  \binom{n}{2} + \frac{nc}{4} - 6\sum_{j=0}^{n-1} h_j \ge 0
\end{equation}
where $V=M_0, \ldots, M_{n-1}$ are a complete list of simple $V$ modules and $h_j$ is the lowest conformal dimension of $M_j$.
Moreover, $\ell$ is an integer.
A strongly rational VOA $V \in (\cC, c)$ with $\rank(\cC) > 1$ is called \emph{extremal} \cite{TenerWang17} if $\sum h_j$ is as large as possible for $c$ in light of \eqref{eqn: ell}.
This is analogous to the extremality condition introduced by H\"{o}hn for holomorphic VOAs (i.e. VOAs $V$ with $\Rep(V) = \oVec$) \cite{Hoehn95}.
Since the $h_j$ are determined mod $1$ by $\cC$, extremality is equivalent to the condition $\ell < 6$.

The classification of (characters of) extremal non-holomorphic VOAs appears to be a tractable piece of the unrestricted general classification problem.
In \cite{TenerWang17}, it was demonstrated that when $\rank(\cC)$ is 2 or 3, then the characters of a VOA are determined by its genus, and a list of potential character vectors was obtained up to central charge $48$.
In this article we give a complete classification of characters of extremal VOAs $V$ with $\rank(\Rep(V)) = 2$, with no restriction on the central charge.

\begin{maintheorem}
There are exactly 15 character vectors of strongly rational extremal (i.e. $\ell < 6$) VOAs with exactly two simple modules.
These characters are listed in Table \ref{tab: characters}.
\end{maintheorem}

This theorem appears in the main body of the text as Theorem \ref{thm: main theorem}.
Of the 15 character vectors in Table \ref{tab: characters}, all but one are realized by previously-known VOAs.
The remaining case, corresponding to the genus $(\Semion, 33)$, is realized by a new VOA constructed in Section \ref{sec: new VOA}.
 
This article fits into a cluster of activity regarding classification of VOAs with two simple modules (or, more generally, two characters).
Just recently, Mason, Nagatomo, and Sakai \cite{MasonNagatomoSakai18ax} used MLDEs to establish a classification result for VOAs with two simple modules satisfying certain additional properties, in the $\ell=0$ regime.
Our results extend the Mason-Nagatomo-Sakai classification to the case $\ell < 6$, although in contrast we only consider the problem of classifying characters of VOAs.
The exceptional $c=33$ character vector has $\ell = 4$, and thus did not appear in earlier classifications.

The Mason-Nagatomo-Sakai classification builds upon work of Franc and Mason  \cite{MasonVVMF08,FrancMason14} which describes solutions to MLDEs in rank 2 in terms of hypergeometric series.
In contrast, our approach is to use the general theory of vector-valued modular forms to derive an explicit recurrence between potential character vectors in the genus $(\cC, c)$ and those in $(\cC, c \pm 24)$.
By studying the long-term behavior of this recurrence, we are able to obtain effective bounds on the possible central charges of extremal VOAs.
While we only consider the case of VOAs with two simple modules in this article, our method does not rely on any special features of rank 2 MLDEs (such as a hypergeometric formula), and in future work we hope to apply the same techniques in higher rank.
For this reason we avoid any explicit use of the hypergeometric series formulas.

Section \ref{sec: 3 characters of VOAs with two simple modules} of this article is an adaptation of the undergraduate thesis \cite{GradyThesis} of the first author, which obtained a classification of characters for extremal VOAs with two simple modules, and which focused on the case $c,h \ge 0$.
Not long after the thesis was published online, an article \cite{ChandraMukhi2019} in the physics literature used MLDEs to obtain a classification similar to the one presented here, without having been aware of \cite{GradyThesis}.

The article is organized as follows.
In Section \ref{sec: 2 rank two categories}, we review the classification of rank two modular tensor categories and modular data from the perspective of VOAs.
In Section \ref{sec: 31 characters and vvmfs}, we review the tools from \cite{BG} which we will use to describe character vectors of VOAs.
In Section \ref{sec: 32 general recurrence}, we derive a recurrence relation which describes how characteristic matrices change under the transformation $c \mapsto c \pm 24$.
In Sections \ref{sec: 33 positive c} and \ref{sec: 34 negative c}, we study the long-term behavior of this recurrence in the positive $c$ and negative $c$ situations, respectively, and in Section \ref{sec: 35 main results} we put these tools together to obtain our main theorem.
Section \ref{sec: new VOA} provides a construction of an extremal VOA in the genus $(\Semion, 33)$.
Finally, in Appendix \ref{sec: A data} we give tables of numerical data used in the proof of the main theorem, as well as all 15 extremal characters in rank two.

\section{Rank two modular tensor categories}\label{sec: 2 rank two categories}

In this article we will consider VOAs which are simple, of CFT type, self-dual, and regular (or equivalently, rational and $C_2$-cofinite \cite{ABD04}).
For brevity, we will use the term \emph{strongly rational} to describe such VOAs.
We refer the reader to \cite{DLM97,ABD04} for background on the adjectives under consideration, but we will explain here the consequences which are relevant for our work.

A strongly rational VOA $V$ possesses finitely many simple modules $V = M_0, M_1, \ldots, M_n$.
We denote the category of $V$-modules by $\Rep(V)$, and write $\rk(\Rep(V))$ for the number of simple modules $n+1$.
We will assume throughout that every module $M_j$ is self-dual, as it simplifies the exposition and is satisfied in the rank two case.

We are primarily interested in the characters of $V$,
$$
\ch_j(\tau) = q^{-c/24} \sum_{n=0}^\infty \dim M_j(n+h_j) \, q^{n+h_j}
$$
where as usual $q=e^{2 \pi i \tau}$, $c$ is the central charge of $V$, $h_j$ is the smallest conformal dimension occurring in $M_j$, and $M_j(n+h_j)$ is the space of states of conformal dimension $n+h_i$.
The foundational work of Zhu \cite{Zhu96} demonstrated that the characters $\ch_j$ define holomorphic functions on the upper half-plane, and that their span is invariant under the action of the modular group.
Thus if we set
$$
\ch(\tau) = \begin{pmatrix} \ch_0 \\ \vdots \\ \ch_n \end{pmatrix},
$$
there exists a representation $\rho_V:\PSL(2,\Z) \to \GL(n+1,\C)$ such that
\begin{equation}\label{eqn: modular invariance}
\ch(\gamma \cdot \tau) = \rho_V(\gamma) \ch(\tau)
\end{equation}
for all $\gamma \in \PSL(2,\Z)$ (recall that we assumed each $M_i$ to be self-dual).
Here $\gamma \cdot \tau$ denotes the natural action of $\PSL(2,\Z)$ on the upper half-plane.

By the work of Huang (\cite{HuangModularity}, see also \cite{Huang05}), $\Rep(V)$ is naturally a modular tensor category, and based on Huang's work Dong-Lin-Ng \cite{DongLinNg15} showed that Zhu's modular invariance is encoded by the $S$ and $T$ matrices of $\Rep(V)$ (see \cite{EGNOBook} for more detail on the $S$ and $T$ matrices of a modular tensor category).
Recall that the normalization of $S$ is only canonical up to a sign, and that for each choice of $S$ the normalization of $T$ is only canonical up to a third root of unitary.
By \cite[Thm. 3.10]{DongLinNg15} (based on \cite{HuangModularity}), we have that $\rho_V\begin{psmallmatrix}0 & 	-1\\ 1 & 0\end{psmallmatrix}$ coincides with a normalization of the categorical $S$ matrix of $\Rep(V)$, and it is straightforward to check directly that 
\begin{equation}\label{eqn: rho T}
\rho_V\begin{psmallmatrix}1 & 1\\ 0 & 1 \end{psmallmatrix} = e^{-2 \pi i c/24}(\delta_{j,k} e^{2 \pi i h_j} )_{j,k}.
\end{equation}

We now consider strongly rational VOAs $V$ such that $\rank(\Rep(V)) = 2$.
We will sometimes write $h$ instead of $h_1$ for the non-trivial lowest conformal dimension, and similarly we will sometimes write $M$ instead of $M_1$.
The classification of modular tensor categories of rank 2 was obtained in \cite{RSW}, and the complete list of normalized $S$ matrices is:
\begin{equation}\label{eqn: possible S with signs}
\pm \frac{1}{\sqrt{2}} \begin{pmatrix} 1 & \varepsilon\\ \varepsilon & -1\end{pmatrix} \qquad \mbox{ and } \qquad \pm \frac{1}{\sqrt{2+\alpha}} \begin{pmatrix} 1 & \alpha\\ \alpha & -1\end{pmatrix}
\end{equation}
where $\varepsilon^2 = 1$ and $\alpha^2 = 1 + \alpha$.
This list coincides with the complete list of $S$ matrices for two-dimensional congruence representations of the modular group earlier obtained by Mason \cite[\S3]{Mason07}.

Observe that 
$$
e^{2\pi i c/24} \, \ch_j(i) =  \sum_{n=0}^\infty \dim M_j(n+h_j) \, e^{-2 \pi(n+h_j)} \,>\, 0
$$ 
and thus the phase of $\ch_j(i)$ is independent of $j$.
By \eqref{eqn: modular invariance}, $\rho\begin{psmallmatrix}0 & 	-1\\ 1 & 0\end{psmallmatrix} \ch(i) = \ch(i)$, and thus $\rho\begin{psmallmatrix}0 & 	-1\\ 1 & 0\end{psmallmatrix}$ fixes a vector all of whose entries have the same phase.
This observation allows us to refine \eqref{eqn: possible S with signs} and conclude that if $\rk(\Rep(V)) = 2$ then $\rho\begin{psmallmatrix}0 & 	-1\\ 1 & 0\end{psmallmatrix}$ must be one of:%
\footnote{%
The authors thank Terry Gannon for pointing out that such a refinement is possible.%
}%
\begin{equation}\label{eqn: possible S}
\frac{1}{\sqrt{2}} \begin{pmatrix} 1 & 1\\ 1 & -1\end{pmatrix},
\quad
\frac{1}{\sqrt{2}} \begin{pmatrix} -1 & 1\\ 1 & 1\end{pmatrix},
\quad
\frac{1}{\sqrt{2+\varphi}} \begin{pmatrix} 1 & \varphi\\ \varphi & -1\end{pmatrix},
\quad
\frac{1}{\sqrt{3-\varphi}} \begin{pmatrix} -1 & \varphi\!-\!1\\ \varphi\!-\!1 & 1\end{pmatrix}
\end{equation}
where we use positive square roots and $\varphi = \frac{1 + \sqrt{5}}{2}$ is the golden ratio.

By the classification of \cite{RSW}, there are exactly two modular tensor categories realizing each of \ref{eqn: possible S} as a normalization of its $S$ matrix, and these two are related by a reversal of the braiding.
Fix one of these 8 modular tensor categories $\cC$ and its normalized $S$-matrix from \eqref{eqn: rho T}.
We wish to see how much information about a hypothetical $V$ with $\Rep(V) = \cC$ we may recover.
By definition, the non-normalized $T$ matrix of $\cC$ is the diagonal matrix $e^{2 \pi i h_j} \delta_{j,k}$, and thus the equivalence class of $h$ mod $1$ is determined by $\cC$.
Observe that if $(S,T)$ are generators of a representation of $\PSL(2,\Z)$ then $(S,\zeta T)$ again generate a representation only if $\zeta^3 = 1$, and thus from \eqref{eqn: rho T} we can see that $c$ mod $8$ is determined by $\cC$ as well.

We summarize the 8 cases in Table \ref{tab: 8 cases}.
Each row corresponds to a modular tensor category, giving its normalized $S$ matrix from \eqref{eqn: possible S}, the equivalence classes of central charge and minimal conformal weight of a hypothetical VOA realization, as well as a familiar name for the category and a VOA realizing the category, where appropriate/known.

\begin{table}[h!]
\renewcommand*{\arraystretch}{1.2}
\begin{tabular}{|c|c|c|c|c|c|}
\hline
\# & $S$ & $c$ mod $8$ & $h$ mod $1$ & Name & Extremal realization\\
\hline
1 & $\frac{1}{\sqrt{2}} \begin{psmallmatrix} 1 & 1\\ 1 & -1\end{psmallmatrix}$& 1 & $\frac14$ & $\Semion$ & $A_{1,1}$ at $c=1$\\
\hline
2 & $\frac{1}{\sqrt{2}} \begin{psmallmatrix} 1 & 1\\ 1 & -1\end{psmallmatrix}$& 7 & $\frac34$ & $\overline{\Semion}$ & $E_{7,1}$ at $c=7$\\
\hline
3 & $\frac{1}{\sqrt{2}} \begin{psmallmatrix} -1 & 1\\ 1 & 1\end{psmallmatrix}$& -3 & $-\frac34$ & $\Semion^\dagger$ & None\\
\hline
4 & $\frac{1}{\sqrt{2}} \begin{psmallmatrix} -1 & 1\\ 1 & 1\end{psmallmatrix}$& -5 & $-\frac14$ & $\overline{\Semion}^\dagger$ & None\\
\hline
5 & $\frac{1}{\sqrt{2+\varphi}} \begin{psmallmatrix} 1 & \varphi\\ \varphi & -1\end{psmallmatrix}$& $\frac{14}{5}$ & $\frac25$ & $\Fib$ & $G_{2,1}$ at $c=\frac{14}{5}$\\
\hline
6 & $\frac{1}{\sqrt{2+\varphi}} \begin{psmallmatrix} 1 & \varphi\\ \varphi & -1\end{psmallmatrix}$& $\frac{26}{5}$ & $\frac35$ & $\overline{\Fib}$ & $F_{4,1}$ at $c=\frac{26}{5}$\\
\hline
7 & $\frac{1}{\sqrt{3-\varphi}} \begin{psmallmatrix} -1 & \varphi\!-\!1\\ \varphi\!-\!1 & 1\end{psmallmatrix} $& $-\frac{22}{5}$ & $-\frac15$ & $\LeeYang$ & $\LeeYang$ at $c=-\frac{22}{5}$\\
\hline
8 & $\frac{1}{\sqrt{3-\varphi}} \begin{psmallmatrix} -1 & \varphi\!-\!1\\ \varphi\!-\!1 & 1\end{psmallmatrix} $& $-\frac{18}{5}$ & $-\frac45$ & $\overline{\LeeYang}$ & None\\
\hline
\end{tabular}
\renewcommand*{\arraystretch}{1}
\smallskip
\caption{The $8$ rank two modular tensor categories from the perspective of VOAs, and an extremal realization where applicable}
  \label{tab: 8 cases}
\end{table}

The genus of a strongly rational VOA $V$ is the pair $(\Rep(V), c)$.
In \cite{TenerWang17}, the third author and Zhenghan Wang defined an \emph{extremal} (non-holomorphic) VOA to be one with $\rank(\Rep(V)) > 1$ and such that the minimal conformal weights $h_j$ were as large as possible in light of a certain a priori bound \cite{MathurMukhiSen88, Mason07}.
See \cite[\S2.2]{TenerWang17} for more detail.
When $\rank(\Rep(V)) = 2$, then $V$ is extremal when
\begin{equation}\label{eqn: extremality}
0 \le 1 + \frac{c}{2} - 6h < 6,
\end{equation}
The quantity $\ell := 1 + \frac{c}{2} - 6h$ is always a non-negative integer, and has been used frequently in the study of VOAs (e.g. \cite{MathurMukhiSen88, MasonNagatomoSakai18ax, GHM} among many others).

The purpose of this article is to provide a list of all possible characters of extremal VOAs with $\rank(\Rep(V)) = 2$.
Given a rank two modular tensor category $\cC$ and central charge $c$ in the appropriate class mod $8$ (as in Table \ref{tab: 8 cases}), there is a unique rational number $h_{ext}$ in the appropriate class mod $1$  satisfying \eqref{eqn: extremality}.
When $\cC$ is fixed we will write $h_{ext}(c)$ to emphasize the dependence on $c$.

The pair $(\cC,c)$ of a modular tensor category and appropriate choice of $c$ is called an \emph{admissible genus} \cite{Hoehn03}.
For every admissible genus $(\cC,c)$ described by Table \ref{tab: 8 cases} there is a representation $\rho_c:\PSL(2,\Z) \to U(2,\C)$ whose $S$ matrix is given by the entry of the table, and whose $T$ matrix is obtained by rescaling the categorical $T$ matrix by $e^{-2 \pi i c/24}$.
These representations are simply a choice of normalization of the categorical $S$ and $T$ matrices, and their existence do not depend in any way on vertex operator algebras.
However, they are defined in such a way that if there is a strongly rational VOA $V$ with central charge $c$ and $\Rep(V) = \cC$, then $\rho_V = \rho_c$.

\newpage

\section{Characters of VOAs with two simple modules} \label{sec: 3 characters of VOAs with two simple modules}

\subsection{Characters and vector-valued modular forms}
\label{sec: 31 characters and vvmfs}

We briefly recall the relevant theory of vector-valued modular forms, following \cite{BG,Gannon14}.
We refer the reader to these references, especially \cite[\S2]{BG}, for more detail.

Let $\rho:\PSL(2,\Z) \to \GL(d,\C)$ be an irreducible representation of the modular group, and assume that $\rho \begin{psmallmatrix} 1 & 1\\ 0 & 1 \end{psmallmatrix}$ is diagonal with finite order.
Let $\bbX:\bbH \to \C$ be a holomorphic function on the upper half-plane which satisfies 
\begin{equation}\label{eqn: modular invariance 2}
\bbX(\gamma \cdot \tau) = \rho(\gamma) \bbX(\tau)
\end{equation}
for all $\gamma \in \PSL(2,\Z)$ and $\tau \in \bbH$.
Choose a diagonal matrix $\Lambda$ such that $\rho \begin{psmallmatrix} 1 & 1\\ 0 & 1 \end{psmallmatrix} = e^{2 \pi i \Lambda}$, called an \emph{exponent matrix}.
For any choice of exponent matrix, we may Fourier expand
\begin{equation}\label{eqn: q expansion}
q^{-\Lambda} \bbX(q) = \sum_{n \in \Z} \bbX[n] q^n
\end{equation}
for coefficients $\bbX[n] \in \C^d$.
Let $\cM(\rho)$ denote the space of functions $\bbX$ satisfying \eqref{eqn: modular invariance 2} such that $\bbX[n] = 0$ for $n$ sufficiently negative (observe that this does not depend on the choice of $\Lambda$).

Given a choice of exponent $\Lambda$, we define the principal part map 
$$
\cP_\Lambda: \cM(\rho) \to \Span \{ v q^{-n} : n > 0, v \in \C^d\}
$$
by
$$
\cP_\Lambda \bbX = \sum_{n < 0} \bbX[n] q^n
$$
where $\bbX[n]$ are as in \eqref{eqn: q expansion}.

An exponent matrix is called \emph{bijective} if $\cP_\Lambda$ is an isomorphism.
For $\xi \in \{1, \ldots, d\}$ let $e_\xi \in \C^d$ be the corresponding standard basis vector.
Given a choice of bijective exponent matrix, let $\bbX^{(\xi)} \in \cM(\rho)$ be the function with $\cP_\Lambda \bbX^{(\xi)} = q^{-1} e_\xi$.
In this case, $\bbX^{(1)}, \ldots, \bbX^{(d)}$ form a basis for $\cM(\rho)$ as a free $\C[J]$-module, where
$$
J = q^{-1} + 196884q + \cdots
$$
is the $J$-invariant.
The \emph{fundamental matrix} $\Xi$ is given by 
$$
\Xi = [ \, \bbX^{(1)}\, | \, \cdots \, | \, \bbX^{(d)} \, ].
$$
The \emph{characteristic matrix} $\chi$ is given by the constant terms of $\Xi$ taken in the $q$-expansion (shifted by $\Lambda$ as in \eqref{eqn: q expansion}).
That is,
$$
\chi = [ \, \bbX^{(1)}[0] \, | \, \cdots \, | \, \bbX^{(d)}[0] \, ]. 
$$

Now fix as in Section \ref{sec: 2 rank two categories} a modular tensor category $\cC$ of rank two, and a choice of real number $c$ in the appropriate class mod $8$.
From this data we specified a representation $\rho_c$ of $\PSL(2,\Z)$ with the property that if there exists a VOA $V$ with central charge $c$ and $\Rep(V) = \cC$, then its character vector $(\ch_j)$ satisfies $(\ch_j) \in \cM(\rho_c)$.
The key observation of \cite{TenerWang17} is that 
$$
\Lambda(c) = \begin{pmatrix} 1-\frac{c}{24} & 0\\ 0 & h_{ext}(c) - \frac{c}{24}\end{pmatrix} =: \begin{pmatrix} \lambda_0(c) & 0 \\ 0 & \lambda_1(c) \end{pmatrix}
$$
is a bijective exponent for $\rho_c$, where $h_{ext}$ is the real number lying in the appropriate class mod $1$ which satisfies \eqref{eqn: extremality}.
Thus by the definition of fundamental matrix we have:

\begin{theorem}[{\cite[Thm. 3.1]{TenerWang17}}] \label{thm: tener wang}
Let $\cC$ be a modular tensor category of rank two, and let $c$ be a real number in the appropriate class mod $8$.
If $V$ is an extremal VOA with central charge $c$ and $\rank(\Rep(V)) = \cC$, then its character appears as the first column of the fundamental matrix corresponding to the bijective exponent $\Lambda(c)$.
\end{theorem}

Let $\chi(c) = (\chi(c)_{ij})_{i,j=0}^1$ be the characteristic matrix taken with respect to $\Lambda(c)$.
Thus if $V$ is an extremal VOA with central charge $c$ (and $\rank(\Rep(V))=2$) we have $\chi(c)_{00} = \dim V(1)$.
We will determine the possible values of $c$ for which there exists an extremal VOA by showing that for $\abs{c}$ sufficiently large, one of $\chi(c)_{00}$ or $\chi(c)_{10}$ is not a non-negative integer.

\subsection{General recurrence}\label{sec: 32 general recurrence}
They key idea \cite{GradyThesis} is to derive a recurrence relating the pair $(\chi(c+24), h_{ext}(c+24))$ to $(\chi(c), h_{ext}(c))$, and then study the long-term behavior of this recurrence.
In fact, to handle the case $c \to +\infty$, one may derive a simple recurrence involving only the diagonal entries of $\chi$ \cite[Lem. 6.4]{GradyThesis}.
To handle the case $c \to -\infty$ we will use all of the entries of $\chi$, and the relation will be slightly more complicated as a result.

Let $M^-_{2 \times 2}$ be the set of $2 \times 2$ complex matrices whose bottom-left entry is non-zero, and let $M^+_{2 \times 2}$ be the set of matrices whose top-right entry is non-zero.
Define functions
$$
f_{\pm}:M^{\mp}_{2\times2} \times (\R \setminus \Z) \to M^{\pm}_{2\times2} \times (\R \setminus \Z)
$$
by
$$
f_+\left[\begin{pmatrix}x & y\\ z & w\end{pmatrix}, \, h\right] = \left[
\renewcommand*{\arraystretch}{1.9}
\begin{pmatrix} \frac{w+h(x-240)}{h+1} & \frac{1}{z} \\
\frac{(h+1)^2(h-2)yz - (x-w + 120(h-1))^2 + 746496(h+1)^2}{(h+2)(h+1)^2} \, z &
\frac{x+h(w+240)}{h+1}
\end{pmatrix}
\renewcommand*{\arraystretch}{1}
, \,\,
h+2 \,\right]
$$
and
$$
f_-\left[\begin{pmatrix}x & y\\ z & w\end{pmatrix}, \, h\right] = \left[
\renewcommand*{\arraystretch}{1.9}
\begin{pmatrix} \frac{-w+(h-2)(x+240)}{h-3} & 
\frac{h(h-3)^2 y z + (x - w + 120(h-1))^2 - 746496(h-3)^2}{(h-4)(h-3)^2} y
\\
\frac{1}{y} &
\frac{-x+(h-2)(w-240)}{h-3}
\end{pmatrix}
\renewcommand*{\arraystretch}{1}
, \,\,
h-2 \,\right].
$$
By direct computation, one may check that these functions are invertible and $f_{\pm}^{-1} = f_{\mp}$.
We will show that $f_{\pm}$ take characteristic matrices to characteristic matrices, but first we must check:

\begin{lemma}\label{lem: characteristic off diagonal non-zero}
Let $\chi$ be the fundamental matrix corresponding to a $2 \times 2$ bijective exponent $\Lambda$.
Then $\chi \in M^+_{2 \times 2} \cap M^-_{2 \times 2}$.
\end{lemma}
\begin{proof}
We show $\chi \in M^-_{2 \times 2}$ and the other step is similar.
Let $\bbX = \bbX^{(1)}$ be the first column of the fundamental matrix $\Xi$ corresponding to $\Lambda = \begin{psmallmatrix} \lambda_{0} & 0\\ 0 & \lambda_{1} \end{psmallmatrix}$.
If the bottom-left entry of $\chi$ were $0$, this would imply that $\bbX$ was of the form
\begin{equation}\label{eqn: expansion off diagonal non-zero}
\bbX = \begin{pmatrix} q^{\lambda_0}(q^{-1} + \cdots) \\ q^{\lambda_1}( z_1 q + \cdots)\end{pmatrix}
\end{equation}
for some $z_1 \in \C$.
By \cite[Thm. 4.1]{Gannon14}, $\Lambda^\prime = \begin{psmallmatrix} \lambda_{0}-1 & 0\\ 0 & \lambda_{1}+1 \end{psmallmatrix}$ is again a bijective exponent.
But examining \eqref{eqn: expansion off diagonal non-zero} we see that $\cP_{\Lambda^\prime} \bbX = 0$, which is a contradiction.
\end{proof}

\begin{lemma}\label{lem: chi recurrence}
Let $(\cC, c)$ be an admissible genus from Table \ref{tab: 8 cases}, and let $\chi(c)$ denote the characteristic matrix of the representation $\rho_c$ taken with respect to the bijective exponent $\Lambda(c)$.
Then $[\chi(c\pm24), h_{ext}(c \pm 24)] = f_\pm[\chi(c),h_{ext}(c)]$.
\end{lemma}
\begin{proof}
It is clear from \eqref{eqn: extremality}, which characterizes $h_{ext}$, that $h_{ext}(c \pm 24) = h_{ext}(c) \pm 2$, and by examining Table \ref{tab: 8 cases} we see that $h_{ext}(c)$ is never an integer.
Since $f_\pm^{-1} = f_{\mp}$ it suffices to show that $[\chi(c+24),h_{ext}(c+24)] = f_+[\chi(c),h_{ext}(c)]$.

Let $h = h_{ext}(c)$.
By definition we have 
$$
\Lambda:=\Lambda(c) = \begin{pmatrix} 1 - \frac{c}{24} &0 \\ 
0 & h - \frac{c}{24} \end{pmatrix}
$$
and
\begin{equation}\label{eqn: lambda plus relation}
\Lambda_+:=\Lambda(c+24) = \begin{pmatrix} - \frac{c}{24} &0 \\ 
0 & h + 1  - \frac{c}{24} \end{pmatrix} = \Lambda(c) + \begin{pmatrix}-1 & 0\\ 0 & 1\end{pmatrix}.
\end{equation}

Let $\Xi:=\Xi(c)$ and $\Xi^+ := \Xi(c+24)$ be the fundamental matrices corresponding to the bijective exponents $\Lambda$ and $\Lambda_+$, respectively, for the representation $\rho_c = \rho_{c+24}$.
We may expand
$$
\Xi = ( \, \bbX^{(1)} \, | \, \bbX^{(2)} \, ) = q^{\Lambda} \begin{pmatrix} q^{-1} + \sum_{n \ge 0} x_n q^n & \sum_{n \ge 0} y_n q^n\\
\sum_{n \ge 0} z_n q^n & q^{-1} + \sum_{n \ge 0} w_n q^n
\end{pmatrix}
$$
and
$$
\Xi_+ = ( \, \bbX_+^{(1)} \, | \, \bbX_+^{(2)} \, ) = q^{\Lambda_+} \begin{pmatrix} q^{-1} + \sum_{n \ge 0} x^+_n q^n & \sum_{n \ge 0} y^+_n q^n\\
\sum_{n \ge 0} z^+_n q^n & q^{-1} + \sum_{n \ge 0} w^+_n q^n
\end{pmatrix}
.
$$
Our goal is to show that 
\begin{equation}\label{eqn: f goal}
f_+\left[\begin{pmatrix} x_0 & y_0\\ z_0 & w_0 \end{pmatrix}, h\right] = \left[\begin{pmatrix} x_0^+ & y^+_0\\ z_0^+ & w_0^+ \end{pmatrix}, h+2\right].
\end{equation}
To do this, we must obtain formulas for $x_0^+$, $y_0^+$, $z_0^+$ and $w_0^+$ in terms of $x_0$, $y_0$, $z_0$, and $w_0$.

Using \eqref{eqn: lambda plus relation}, we have
\begin{equation}\label{eqn: Xi early terms}
\Xi_+ = q^{\Lambda} \begin{pmatrix}
q^{-2}+x_0^+q^{-1}+ \cdots & y_0^+q^{-1}+ \cdots\\
z_0^+ q + \cdots  & 1 + \cdots 
\end{pmatrix}.
\end{equation}
Thus 
$$
\cP_{\Lambda} \, \bbX^{(1)}_+ = \begin{pmatrix} q^{-2} + x_0^+ q^{-1}\\ 0 \end{pmatrix}.
\quad \mbox{ and } \quad
\cP_{\Lambda} \, \bbX^{(2)}_+ = \begin{pmatrix} y_0^+ q^{-1} \\ 0 \end{pmatrix}.
$$
Since $\cP_{\Lambda}$ is injective, we have 
\begin{equation}\label{eqn: X plus 2 formula}
\bbX_+^{(2)} = y_0^+ \bbX^{(1)}.
\end{equation}
By \eqref{eqn: Xi early terms}, this identity reads
$$
\begin{pmatrix}
y_0^+ q^{-1} + \cdots\\
1 + \cdots
\end{pmatrix}
=
\begin{pmatrix}
y_0^+ q^{-1} + \cdots\\
y_0^+z_0 + \cdots
\end{pmatrix}
$$
and thus 
\begin{equation}\label{eqn: y0p formula}
y_0^+ = \frac{1}{z_0}.
\end{equation} 
Note that $z_0 \ne 0$ by Lemma \ref{lem: characteristic off diagonal non-zero}.
This gives the formula for $y_0^+$ which corresponds to \eqref{eqn: f goal}.

Substituting \eqref{eqn: y0p formula} into \eqref{eqn: X plus 2 formula} yields $\bbX_+^{(2)} = \frac{1}{z_0} \bbX^{(1)}$, from which we conclude 
\begin{equation}\label{eqn: w0p formula}
w_0^+ = \frac{z_1}{z_0}
\end{equation}
by considering the second $q$-coefficient in the second entry of $\bbX_+^{(2)}$.
While this gives an expression for $w_0^+$ in terms of $z$ variables, it is not the expression we are looking for due to the presence of the higher order coefficient $z_1$.
We will derive an expression for $z_1$ in terms of lower order coefficients later in the proof \eqref{eqn: z1 formula}.

For now, we continue on and find expressions for $x_0^+$ and $z_0^+$.
By direct calculation,
$$
\cP_{\Lambda}\big((J + x_0^+ - x_0)\bbX^{(1)} - z_0 \bbX^{(2)}\big) = \begin{pmatrix} q^{-2} + x_0^+ q^{-1}\\ 0 \end{pmatrix} = \cP_{\Lambda}\bbX^{(1)}_+
$$
where $J(q) = q^{-1} + 196884q + \cdots$.
Since $\cP_{\Lambda}$ is injective, we have 
\begin{equation}\label{eqn: first column}
(J + x_0^+ - x_0)\bbX^{(1)} - z_0 \bbX^{(2)} = \bbX^{(1)}_+.
\end{equation}
We now multiply both sides of \eqref{eqn: first column} by $q^{-\Lambda}$, expanding out the left-hand side and substituting the expression from \eqref{eqn: Xi early terms} for the right-hand side, to obtain
$$
\begin{pmatrix}
q^{-2} + x_0^+ q^{-1} + \cdots\\
\gamma_0 + \gamma_1 q + \cdots
\end{pmatrix}
=
\begin{pmatrix}
q^{-2}+x_0^+q^{-1}+ \cdots \\
0 + z_0^+ q + \cdots 
\end{pmatrix}
$$
where
$$
\gamma_0 = (x_0^+-w_0-x_0) z_0 + z_1,
\quad \mbox{ and } \quad
\gamma_1 = 196884 z_0 - w_1 z_0 - x_0 z_1 + x_0^+ z_1 + z_2.
$$
Thus $\gamma_0 = 0$ and $\gamma_1 = z_0^+$.
The former yields
\begin{equation}\label{eqn: x0p formula}
x_0^+ = 
x_0 + w_0 -\frac{z_1}{z_0}.
\end{equation}
Substituting \eqref{eqn: x0p formula} into the equation $z_0^+ = \gamma_1$ and simplifying yields
\begin{equation}\label{eqn: z0p formula}
z_0^+ = -\frac{z_1^2}{z_0} + z_1 w_0 + z_0(196884 - w_1) + z_2.
\end{equation}

Our aim now is to replace the higher order coefficients $w_1$, $z_1$, and $z_2$ appearing in \eqref{eqn: w0p formula}, \eqref{eqn: x0p formula} and \eqref{eqn: z0p formula} with expressions in terms of $x_0$, $y_0$, $z_0$, and $w_0$.
To do this, we use the differential equation \cite[Eqn. (2.14)]{BG}
\begin{equation}\label{eqn: Xi differential equation}
\frac{1}{2\pi i}\frac{d \Xi}{d \tau} - \Xi(\tau) \cD(\tau) = 0
\end{equation}
where
$$
\cD(\tau) = \frac{1}{\cE(\tau)}\big[ (J(\tau) - 240)(\Lambda - 1) + \chi + [\Lambda, \chi]\big]
$$
for $\cE(\tau) = q^{-1} - 240 - 141444q - \cdots$.
Examining the coefficient of $q$ in the bottom-right entry of \eqref{eqn: Xi differential equation} yields
\begin{equation}\label{eqn: w1 formula}
w_1 = \frac12\big(w_0 (w_0 + 240) - (h - 2) y_0 z_0 + 338328 (h - 1 - \frac{c}{24})\big).
\end{equation}
Similarly, examining the coefficients of $q$ and $q^2$ in the bottom-left entry of \eqref{eqn: Xi differential equation} yields
\begin{equation}\label{eqn: z1 formula}
z_1 = \frac{x_0 + h(w_0 + 240)}{h+1} z_0
\end{equation}
and
\begin{equation}\label{eqn: z2 formula}
z_2 = \frac{x_0(z_1+240z_0) + z_0(hw_1 + 240hw_0 + 199044h - 338328\frac{c}{24})}{h+2}
\end{equation}
respectively.

We have now obtained expressions \eqref{eqn: w1 formula}, \eqref{eqn: z1 formula}, and \eqref{eqn: z2 formula} for $w_1$, $z_1$, and $z_2$ in terms of lower order coefficients.
We may substitute these formulas into our expressions \eqref{eqn: w0p formula}, \eqref{eqn: x0p formula}, and \eqref{eqn: z0p formula} for $w_0^+$, $x_0^+$, and $z_0^+$ to obtain the formulas of \eqref{eqn: f goal}.
Combining with our earlier expression \eqref{eqn: y0p formula} for $y_0^+$ completes the proof.
\end{proof}

\subsection{Recurrence for large positive $c$}\label{sec: 33 positive c}

We will show that for $n$ sufficiently large, $\chi(c+24n)_{00} < 0$, and moreover we will obtain an effective bound on such an $n$.
We will do this by iterating $f_+$, although in fact a simpler function will suffice.

\begin{lem}\label{lem: diagonal recurrence}
Let $g:\C^2 \times (\R \setminus \Z) \to \C^2 \times (\R \setminus \Z)$ be the function 
$$
g[x,w,h] = \left[\frac{w+h(x-240)}{h+1}, \, \frac{x+h(w+240)}{h+1}, h+2 \right],
$$
and let $g^n$ denote its $n$-fold iterate.
Then
$$
g^n[x,w,h] = \left[\frac{n w+(h+n-1)(x-240n)}{h+2n-1}, \, \frac{n x+(h+n-1)(w+240n)}{h+2n-1}, \, h+2n\right]
$$
\end{lem}
\begin{proof}
This follows by a straightforward induction.
\end{proof}

\begin{lem}\label{lem: positive c recurrence}
Let $(\cC, c)$ be an admissible genus from Table \ref{tab: 8 cases}, and let $\chi(c)$ denote the characteristic matrix of the representation $\rho_c$ taken with respect to the bijective exponent $\Lambda(c)$.
Suppose that $h_{ext}(c) > 0$.
Then $\chi(c+24n)_{00} < 0$ when 
$$
n > \frac{\abs{M} + \sqrt{M^2+960\abs{(h_{ext}(c)-1) \chi(c)_{00}}}}{480}
$$
where $M = \chi(c)_{00} + \chi(c)_{11} - 240(h_{ext}(c)-1)$.
\end{lem}
\begin{proof}
Set $a = \chi(c)_{00}$, $d = \chi(c)_{11}$, and $h = h_{ext}(c)$.
By Lemma \ref{lem: chi recurrence} and Lemma \ref{lem: diagonal recurrence}, we have
$$
\chi(c+24n)_{00} = \frac{n d+(h+n-1)(a-240n)}{h+2n-1}.
$$
Since we assume $h > 0$, when $n \ge 0$ we have $h + 2n - 1 > 0$.
Thus $\chi(c+24n)_{00} < 0$ if and only if 
\begin{equation}\label{eqn: plus quadratic}
0 > n d+(h+n-1)(a-240n) = -240n^2 + M n + (h-1) \chi_{00}.
\end{equation}
The right-hand side of \eqref{eqn: plus quadratic} is a quadratic polynomial in $n$ which is concave down.
Thus \eqref{eqn: plus quadratic} holds when $n$ exceeds the largest real root of that quadratic (and it holds trivially if the quadratic has no real roots).
The conclusion of the lemma now follows immediately from the quadratic formula.
\end{proof}

The purpose of Lemma \ref{lem: positive c recurrence} is to reduce the question of classifying extremal VOAs to a finite one.
We apply it 24 times to obtain the following.

\begin{samepage}
\begin{theorem}\label{thm: cmax}
For every rank two modular tensor category $\cC$, there is an explicitly computable $c_{max}$ such that there are no extremal VOAs in the genus $(\cC,c)$ when $c > c_{max}$.
The values are given in Table \ref{tab: table cmax}.
\begin{table}[!htbp]
\begin{equation*}
\renewcommand*{\arraystretch}{1.1}
\begin{array}{|c|c|c|}
\hline
\# & \cC & c_{max}\\
\hline
 \,\,1 \,\, & \,\,\Semion \,\, & \,\, 57 \,\,\\
\hline
 2  & \overline{\Semion}  &  39\\
\hline
 3  & \Semion^\dagger  &  67\\
\hline
 4  & \overline{\Semion}^\dagger  &  37\\
\hline
 5  & \Fib  &  \frac{174}{5}\\
\hline
 6  & \overline{\Fib}  &  \frac{186}{5}\\
\hline
 7  & \LeeYang  &  \frac{338}{5}\\
\hline
 8  & \overline{\LeeYang}  &  \frac{222}{5}\\
\hline
\end{array}
\renewcommand*{\arraystretch}{1}
\end{equation*}
\caption{$c_{max}$ values for rank two modular tensor categories}
\label{tab: table cmax}
\end{table}
\FloatBarrier
\end{theorem}
\end{samepage}
\begin{proof}
Let us first take $\cC$ to be the Semion MTC.
In this case, $c \equiv 1$ mod $8$.
We consider first the case $c \equiv 1$ mod $24$.
For $c=1$, we can compute the characteristic matrix $\chi(1) = \begin{pmatrix} 3 & 26752\\ 2 & -247 \end{pmatrix}$, using for example using the method of \cite{TenerWang17} (based on \cite{BG}), or the method of hypergeometric series \cite{FrancMason14}.
We can compute $h_{ext}(1) = \frac14$ from the definition of $h_{ext}$ and the fact that $h \equiv \frac14$ mod $1$.
Applying Lemma \ref{lem: positive c recurrence} with this data, we see that $\chi(1+24n) < 0$ when $n > 0.298\ldots$.
Thus if $n_{max} = 0$, we have $\chi(1+24n) < 0$ when $n > n_{max}$.
By Theorem \ref{thm: tener wang}, there are no extremal VOAs in the genera $(\cC, 1+24n)$ when $n > n_{max}$.

We can repeat the above exercise for the values $c=9$ and $c=17$, and three times again for each row of Table \ref{tab: 8 cases}.
The resulting characteristic matrices, $h_{ext}$, and $n_{max}$ are given in Table \ref{tab: nmax}.
For each category $\cC$, the value $c_{max}$ in Table \ref{tab: table cmax} is the maximum of the three values of $c + 24 n_{max}$, corresponding to the three possible classes of $c$ mod $24$.
\end{proof}

\subsection{Recurrence for very negative $c$}\label{sec: 34 negative c}

We will show that for $n$ sufficiently large, we have $\abs{\chi(c-24n)_{10}} < 1$.
Since $\chi(c-24n)_{10} \ne 0$ by Lemma \ref{lem: characteristic off diagonal non-zero}, this will guarantee that $\chi(c-24n)_{10}$ is not an integer.
As with the case of very positive $c$, we will avoid finding an explicit expression for $f_{\pm}^n[\chi(c),h_{ext}(c)]$.
Instead, we extract the following pieces of the data which will be easier to work with.
Let $\alpha(c) = \chi(c)_{00} - \chi(c)_{11}$ and $\beta(c) = \chi(c)_{10} \chi(c)_{01}$.

The utility of studying $\beta(c)$ is the following.
\begin{lemma}\label{lem: beta controls chi}
Let $(\cC,c)$ be an admissible genus from Table \ref{tab: 8 cases}, and suppose that $\abs{\chi(c)_{10}} \le 1$ and $\abs{\beta(c-24n)} > 1$ for all $n \ge 1$.
Then $\abs{\chi(c-24n)_{10}} < 1$ for all $n \ge 0$.
\end{lemma}
\begin{proof}
By Lemma \ref{lem: chi recurrence} we have $\chi(c-24)_{10} = \chi(c)_{01}^{-1}$.
Thus 
$$
\beta(c) = \chi(c)_{10}\chi(c)_{01} = \frac{\chi(c)_{10}}{\chi(c-24)_{10}}.
$$
Thus if we know that $\abs{\chi(c)_{10}} \le 1$ and $\abs{\beta(c-24)} > 1$,  we can conclude that $\abs{\chi(c-24)}_{10} < 1$.
We repeat this argument $n$ times to complete the proof.
\end{proof}

\begin{samepage}
The see how $\beta(c)$ depends on $c$, we introduce the function
$$
k:\C^2 \times (\R \setminus \Z) \to \C^2 \times (\R \setminus \Z)
$$
given by
$$
k[\alpha, \beta, h] = \left[\frac{\alpha(h-1) + 480(h-2)}{h-3} ,\,
\frac{ (h-3)^2 (h\beta - 746496) + (\alpha + 120(h-1))^2 }{(h-4)(h-3)^2}  , \, 
  h-2\right].
$$
\end{samepage}
This function was chosen so that:
\begin{lemma}\label{lem: beta evolution}
Let $(\cC, c)$ be an admissible genus from Table \ref{tab: 8 cases}.
Then 
$$
[\alpha(c-24), \beta(c-24), h_{ext}(c-24)] = k[\alpha(c), \beta(c), h_{ext}(c)].
$$
\end{lemma}
\begin{proof}
This follows by direct algebraic manipulation applied to Lemma \ref{lem: chi recurrence} and the formula for $f_-$.
\end{proof}

It is now an algebra exercise to determine the long-term behavior of $\alpha(c-24n)$ and $\beta(c-24n)$.

\begin{lemma}\label{lem: k induction}
The $n$-fold iterate of $k$ is given by
\begin{align*}
k^n&[\alpha,\beta,h] =\\
 & \left[ \frac{ \alpha(h-1) + 480n(h-n-1)}{h-2n-1} , \right.\\
& \quad \frac{n(h-n-1)(\alpha+120(h-1))^2}{(h-2n)(h-2n-1)^2(h-2n-2)} + \frac{h(h-2) \beta - 746496n(h-n-1)}{(h-2n)(h-2n-2)},\\
& \quad h-2n 
\mathlarger{\mathlarger{\mathlarger{\mathlarger{]}}}}
\end{align*}
\end{lemma}
\begin{proof}
The formula may be verified by a straightforward induction using the definition of $k$.
\end{proof}

We carefully examine the expression obtained in Lemma \ref{lem: k induction} to obtain a criterion to bound $\abs{\beta(c+24n)} > 1$.

\begin{lemma}\label{lem: negative c bound}
Let $(\cC, c)$ be an admissible genus from Table \ref{tab: 8 cases}, and suppose that $h_{ext}(c) < 0$ and $\beta(c) > 1$.
Then $\abs{\beta(c-24n)} > 1$ whenever
\begin{equation}\label{eqn: ugly assumption}
n > \frac{\abs{\alpha(c) - 120(1-h_{ext}(c))}(1-h_{ext}(c))}{860}.
\end{equation}
\end{lemma}
\begin{proof}
Let $\beta_n = \beta(c-24n)$, which by Lemma \ref{lem: beta evolution} and Lemma \ref{lem: k induction} is given by the formula
$$
\beta_n = \frac{n(h-n-1)(\alpha+120(h-1))^2}{(h-2n)(h-2n-1)^2(h-2n-2)} + \frac{h(h-2) \beta - 746496n(h-n-1)}{(h-2n)(h-2n-2)}.
$$
where $\alpha = \alpha(c)$ and $h = h_{ext}(c) < 0$.

We can write $\beta_n = \frac{p(n)}{q(n)}$ for $q(n) = (h-2n)(h-2n-1)^2(h-2n-2)$ and $p$ a certain polynomial of $n$.
To show $\abs{\beta_n} > 1$, it suffices to show that $\abs{p(n)} > \abs{q(n)}$.
Since $q(n) > 0$ by inspection when $n \ge 1$, it suffices to show that $p(n) > q(n)$, or that $r(n):= p(n) - q(n) > 0$.
Through straightforward manipulation of the formula for $\beta_n$, we have $r(n) = r_1(n) + r_2(n)$ where
\begin{align*}
r_1(n) = &2985968 n^4 + 5971936 (1 - h) n^3 +\\
&+ (3732456 (1 - h)^2 + 
    4) n^2 + (746488 (1 - h)^3 + 4 (1 - h)) n +\\
&+ 4 (-h) n (2 - h) (1 - h + n)+ (-h) (2 - h) (\beta - 1) (1 - h + 2 n)^2\\
r_2(n) = &-(1 - h + n) n (\alpha - 120 (1 - h))^2
\end{align*}
Since $h < 0$ and $\beta > 1$, every term of $r_1(n)$ is positive.
Thus to show $r(n) > 0$ it suffices to find a term of $r_1(n)$ which controls $r_2(n)$.

To show 
$$
2985968n^4 + r_2(n) > 0,
$$ 
it suffices to show 
$$
2985968n^4 > (1-h+n)^2 (\alpha - 120(1-h))^2.
$$
This will follow from the simple estimate Lemma \ref{lem: silly estimate} below with $A = 2985968$, $B = (\alpha - 120(1-h))^2$, and $C = 1-h$, provided
$$
2985968n^2 > 2(\alpha - 120(1-h))^2(1+(1-h)^2).
$$
This would follow from 
$$
2985968n^2 > 4(\alpha - 120(1-h))^2(1-h)^2,
$$
or equivalently
$$
(2985968)^{\tfrac12} n > 2\abs{\alpha - 120(1-h)}(1-h).
$$
This is an immediate consequence of our assumption \eqref{eqn: ugly assumption}.
\end{proof}

We used the following simple observation in the proof of Lemma \ref{lem: negative c bound}.

\begin{lemma}\label{lem: silly estimate}
Let $A$,$B$, $C$, and $n$ be positive real numbers with $n \ge 1$.
Then if $$An^2 > 2B(1+C^2)$$ it follows that $$An^{4} > B(n+C)^2.$$
\end{lemma}
\begin{proof}
It suffices to show $An^{4} > 2B(n^2 + C^2)$, or equivalently $(An^2-2B)n^2 > 2BC^2$.
Instead, we may show $An^2-2B > 2BC^2$ since $n \ge 1$ and $An^2 > 2B$.
This follows immediately from our hypothesis.
\end{proof}

We now apply Lemma \ref{lem: negative c bound} in 24 cases to obtain a lower bound on the central charge of extremal VOAs.

\begin{samepage}

\begin{theorem}\label{thm: cmin}
For every rank two modular tensor category $\cC$, there is an explicitly computable $c_{min}$ such that there are no extremal VOAs in the genus $(\cC,c)$ when $c < c_{min}$.
The values are given in Table \ref{tab: table cmin}.
\FloatBarrier
\begin{table}[!htbp]
\begin{equation*}
\renewcommand*{\arraystretch}{1.1}
\begin{array}{|c|c|c|}
\hline
\# & \cC & c_{min}\\
\hline
 \,\,1 \,\, & \,\,\Semion \,\, & \,\, -23 \,\,\\
\hline
 2  & \overline{\Semion}  &  -17\\
\hline
 3  & \Semion^\dagger  &  -13\\
\hline
 4  & \overline{\Semion}^\dagger  &  -19\\
\hline
 5  & \Fib  &  -\frac{106}{5}\\
\hline
 6  & \overline{\Fib}  &  -\frac{94}{5}\\
\hline
 7  & \LeeYang  &  -\frac{62}{5}\\
\hline
 8  & \overline{\LeeYang}  &  -\frac{98}{5}\\
\hline
\end{array}
\renewcommand*{\arraystretch}{1}
\end{equation*}
\caption{$c_{min}$ values for rank two modular tensor categories}
\label{tab: table cmin}
\end{table}
\end{theorem}
\end{samepage}
\begin{proof}
As in the proof of Theorem \ref{thm: cmax}, we will work through the necessary computation when $\cC = \Semion$ and obtain a bound which holds for $c \equiv 1$ mod $24$.
Since $h_{ext}(1) > 0$, we must instead consider $c=-23$ in order to apply Lemma \ref{lem: negative c bound}.
We compute $h_{ext}(-23) = -\frac{7}{4}$ using \eqref{eqn: extremality}, and we compute 
$$
\chi(-23) = \begin{pmatrix}
\frac{713}{11} & \frac{57264144384}{11}\\
\frac{1}{26752} & -\frac{3397}{11}
\end{pmatrix}
$$
and from there $\alpha(-23) = \frac{4110}{11}$ and $\beta(-23) = \frac{23546112}{121}$.
Thus by Lemma \ref{lem: negative c bound} we have $\abs{\beta(-23-24n)} > 1$ when $n > 0.13\ldots$.
Taking $n_{max} = 0$, we have $\abs{\beta(-23-24n)} > 1$ when $n > n_{max}$.
As $\abs{\chi(-23-24n_{max})_{10}} < 1$, we conclude that $\abs{\chi(-23-24n)_{10}} < 1$ for all $n > n_{max}$, and thus by Theorem \ref{thm: tener wang} there cannot be an extremal VOA in the genus $(\cC, c)$ when $c < -23$ and $c \equiv 1$ mod $24$. 
We repeat this argument for the other two equivalence classes of $c$ mod $24$, and the value $c_{min}$ from Table \ref{tab: table cmin} is the minimum of the allowed values.

We apply the above procedure to each of the $8$ modular categories appearing in Table \ref{tab: 8 cases}.
The data from each of the cases is given in Table \ref{tab: nmax negative}.
\end{proof}

\subsection{Main result}\label{sec: 35 main results}

Combining Theorem \ref{thm: cmax} and Theorem \ref{thm: cmin}, we obtain for every rank two modular tensor category $\cC$ a pair of numbers $c_{min}$ and $c_{max}$ such that if $V$ is an extremal VOA in the genus $(\cC,c)$, then $c_{min} \le c \le c_{max}$.
We can now compute the characteristic matrix of every remaining pair $(\cC,c)$ (e.g. by Lemma \ref{lem: chi recurrence}), and throw away any for which the first column does not consist of positive integers.
We are left with 15 possibilities, all but one of which are realized by VOAs which have previously been studied.
The remaining character vector is realized by a VOA constructed in Section \ref{sec: new VOA}.
We summarize the result here:

\newpage

\begin{samepage}
\begin{theorem}\label{thm: main theorem}
Let $V$ be a strongly rational extremal VOA with two simple modules.
Then it lies in one of the genera listed in Table \ref{tab: genera of extremal} (and its character vector is given in Table \ref{tab: characters}).

\FloatBarrier
\begin{table}[!htbp]
\begin{minipage}[t]{0.45\textwidth}
$$
\renewcommand*{\arraystretch}{1.3}
\begin{array}{|c|c|c|c|c|}
\hline
\cC & c & \mbox{Realization} & h_{ext} & \ell\\
\hline
\Semion & 1 & A_{1,1} &
\frac14 & 0
\\
\hline
\Semion & 9 & A_{1,1} \otimes E_{8,1} &
\frac14 & 4\\
\hline
\Semion & 17 & \mbox{\cite{GHM}} &
\frac54 & 2\\
\hline
\Semion & 33 & \S \ref{sec: new VOA} &
\frac94 & 4\\
\hline
\overline{\Semion} & 7 & E_{7,1} &
\frac34 & 0\\
\hline
\overline{\Semion} & 15 & E_{7,1} \otimes E_{8,1} &
\frac34 & 4\\
\hline
\overline{\Semion} & 23 & \mbox{\cite{GHM}} &
\frac74 & 2\\
\hline
\end{array}
$$
\end{minipage}
\quad
\begin{minipage}[t]{0.45\textwidth}
$$
\renewcommand*{\arraystretch}{1.3}
\begin{array}{|c|c|c|c|c|}
\hline
\cC & c & \mbox{Realization} & h_{ext} & \ell\\
\hline
\Fib & \frac{14}{5} & G_{2,1} &
\frac25 & 0\\
\hline
\Fib & \frac{54}{5} & G_{2,1} \otimes E_{8,1} &
\frac25 & 4\\
\hline
\Fib & \frac{94}{5} & \mbox{\cite{GHM}} &
\frac75 & 2\\
\hline
\overline{\Fib} & \frac{26}{5} & F_{4,1} &
\frac35 & 0\\
\hline
\overline{\Fib} & \frac{66}{5} & F_{4,1} \otimes E_{8,1} &
\frac35 & 4\\
\hline
\overline{\Fib} & \frac{106}{5} & \mbox{\cite{GHM}} &
\frac85 & 2\\
\hline
\LeeYang & -\frac{22}{5} & \LeeYang &
-\frac15 & 0\\ 
\hline
\LeeYang & \frac{18}{5} & \operatorname{Y-L} \otimes E_{8,1}&
-\frac15 & 4\\ 
\hline
\end{array}
\renewcommand*{\arraystretch}{1}
$$
\end{minipage}
\caption{Genera of extremal VOAs corresponding to rank two modular tensor categories}
\label{tab: genera of extremal}
\end{table}
\renewcommand*{\arraystretch}{1}
\end{theorem}
\end{samepage}
\bigskip

One of the main purposes of establishing classification results such as Theorem \ref{thm: main theorem} is to find interesting new examples, such as the VOA in the genus $(\Semion, 33)$ constructed in Section \ref{sec: new VOA}.
The most interesting genera for which no VOA realizations are known are $(\DHaag,8n)$, where $\DHaag$ is the double of the Haagerup fusion category.
Evans and Gannon computed possible character vectors for potential ``Haagerup VOAs'' in the cases $n=1,2,3$ \cite{EvansGannon11}, and used these characters to suggest strategies for constructing them.
Subsequently Gannon analyzed the possible Lie algebra structures on the weight one vectors of Haagerup VOAs, as well as the structure of their cosets.
Despite all of the circumstantial evidence, however, no construction has been found for a Haagerup VOA.

The success in constructing a $(\Semion,33)$ VOA may be regarded as further evidence of the fruitfulness of the Evans-Gannon approach to the Haagerup VOA.
While the $\Semion$ category is quite a bit simpler than $\DHaag$, the central charge $c=33$ is in practice quite large compared to $c=8,16$, or $24$, which is a source of added difficulty.

In the case of the $(\Semion,33)$ VOA, the subVOA generated by weight 1 vectors is of type $A_{1,1}$.
The coset of this affine VOA is of independent interest, and we record here its character vector
$$
q^{-32/24}
\begin{pmatrix}
1 + 0q + 69616q^2 + 34668544q^3 + \cdots\\
q^{9/4}(426192 + 121366368q + \cdots)\\
q^{7/4}(10245 + 11330970q + \cdots)\\
q^{2}(69888 + 34664448q + \cdots)\\
\end{pmatrix}
$$
as well as the character vector of its holomorphic extension (the twisted orbifold of the rank 32 Barnes-Wall lattice):
$$
q^{-32/24}(1 + 0q + 139504q^2 + 69332992q^3 + \cdots).
$$

\newpage

\section{Construction of the extremal $c=33$ example}\label{sec: new VOA}

\medskip

\baselineskip 5.2mm

The goal of this section is to prove the following theorem.

\begin{theorem}\label{thm:0}
There exists an extremal VOA in the genus $(\Semion, 33)$.
\end{theorem}

The key step in the construction will be the following.

\begin{theorem}\label{thm:1}
There exists a $c=32$ holomorphic framed VOA $V$ and its involution $\theta\in \mathrm{Aut}(V)$ 
satisfying the following conditions.
\\
\textup{(1)}~ $V(1)=0$.
\\
\textup{(2)}~ The unique irreducible $\theta$-twisted $V$-module $W$ has top weight $7/4$.
\end{theorem}

We first give a proof of Theorem \ref{thm:0} using Theorem \ref{thm:1}.\\

\paragraph{\textbf{Proof of Theorem \ref{thm:0}:}}

Suppose we have $V$, $\theta$ and $W$ as in Theorem \ref{thm:1}. 
Let $V^\pm=\{ a\in V \mid \theta a=\pm a\}$ be the eigenspace decompositions and 
$W=W^+\oplus W^-$ the irreducible decomposition as a $V^+$-module.
Note that $V^+$ is a strongly rational VOA by \cite{DongGriessHoehn98} (see also \cite{CarnahanMiyamoto16}).
We assign the labeling $W^\pm$ such that the conformal weight of $W^+$ is 7/4. 
It turns out that the conformal weight of $W^-$ is equal to $9/4$ 
so that the conformal weights of $V^+$, $V^-$, $W^+$ and $W^-$ are 
respectively 0, 2, 7/4 and 9/4.
Since $V$ is holomorphic, $V^+$ has exactly four irreducible modules $V^\pm$ and $W^\pm$, 
all of them are simple currents.
The fusion algebra of $V^+$ is isomorphic to the group algebra associated with $(\Z/2\Z)^2$.

We consider a tensor product $V^+\otimes L_{\hat{\mathfrak{sl}}_2}(1,0)$ and its 
$\Z/2\Z$-graded simple current extension 
\[
  U=V^+\otimes L_{\hat{\mathfrak{sl}}_2}(1,0)\oplus W^+\otimes L_{\hat{\mathfrak{sl}}_2}(1,1),
\]
which is strongly rational by \cite{Lam01,Yamauchi04} (see also \cite[Thm. 4.13]{McRae19}).
Then the weight one subspace of $U$ is 3-dimensional.
It is easy to see that $U$ has exactly two irreducible untwisted modules, $U$ and
\[
  M=V^-\otimes L_{\hat{\mathfrak{sl}}_2}(1,1)\oplus W^-\otimes L_{\hat{\mathfrak{sl}}_2}(1,0), 
\]
whose conformal weight is $2+1/4=9/4+0=9/4$.
Thus $U$ is an extremal VOA with central charge $33$ and two simple modules, as desired.
\qed

We will now prove Theorem \ref{thm:1} based on the theory of framed VOAs.
We will use the same notation and terminology for framed VOAs as in \cite{LamYamauchi08}.
In particular,  the product of codewords is defined by
\[
  \alpha \cdot \beta=(\alpha_1\beta_1,\dots,\alpha_n\beta_n) \in \Z_2^n
\]
for $\alpha=(\alpha_1,\dots,\alpha_n)$ and $\beta =(\beta_1,\dots,\beta_n)\in \Z_2^n$.
We denote by $\mathbf{1}_n$ the codeword $(1^n)\in \Z_2^n$.

Let $V$ be a framed VOA with a Virasoro frame 
$F=\langle e^1,\dots,e^n\rangle\cong L(\sfr{1}{2},0)^{\otimes n}$.
Let $(C,D)$ be the structure codes with respect to $F$ and $V=\oplus_{\alpha\in D}V^\alpha$ 
the corresponding 1/16-word decomposition.
For $\beta=(\beta_1,\dots,\beta_n)\in \Z_2^n$, we define
\begin{equation}\label{eq:1}
  \sigma_{\beta}:=\prod_{i\in \mathrm{supp}(\beta)} {(-1)^{2\mathrm{o}(e^i)}},~~~~
  \tau_{\beta}:=\prod_{i\in \mathrm{supp}(\beta)} {(-1)^{16\mathrm{o}(e^i)}},
\end{equation}
where $\mathrm{o}(a)$ denotes the grade preserving operator of $a\in V$.
It follows from the fusion rules of $L(\sfr{1}{2},0)$-modules that 
$\sigma_\beta\in \mathrm{Aut}(V^0)$ and $\tau_\beta\in \mathrm{Aut}(V)$ (cf.~\cite{Miyamoto96Griess}).
The maps $\sigma :\Z_2^n \to \mathrm{Aut}(V^0)$ and $\tau: \Z_2^n \to \mathrm{Aut}(V)$ are group 
homomorphisms such that $\ker \sigma=C^\perp$ and $\ker \tau=D^\perp$.
Therefore we have the following exact sequences
\begin{equation}\label{eq:2}
\begin{array}{l}
  0\longrightarrow C^\perp \longrightarrow \Z_2^n \xrightarrow{~\sigma~} \mathrm{Im}\, \sigma
  \longrightarrow 0,
  \medskip\\
  0\longrightarrow D^\perp \longrightarrow \Z_2^n \xrightarrow{~\tau~} \mathrm{Im}\, \tau 
  \longrightarrow 0.
\end{array}
\end{equation}
We define the point-wise frame stabilizer\footnote{%
This group was denoted by $\mathrm{Stab}_V^{\mathrm{pt}}(F)$ in \cite{LamYamauchi08}.}
by
\begin{equation}\label{eq:23}
  \mathrm{Aut}_F(V):=\{ g\in \mathrm{Aut}(V) \mid g(e^i)=e^i~\mbox{ for } 1\leq i\leq n\} .
\end{equation}
The structure of the point-wise frame stabilizer is determined in \cite{LamYamauchi08} as follows.

\begin{theorem}[\cite{LamYamauchi08}]\label{thm:2}
Let $V$ be a framed VOA with structure codes $(C,D)$.
\\
\textup{(1)}
$\mathrm{Im}\,\tau \cong \Z_2^n/D^\perp$ is a central subgroup of $\mathrm{Aut}_F(V)$.
\\
\textup{(2)} 
For $\theta\in \mathrm{Aut}_F(V)$, there exists $\xi \in \Z_2^n$ and $\eta\in \Z_2^n$ such that
$\theta|_{V^0}=\sigma_\xi$ and $\theta^2=\tau_\eta$.
In particular, $\theta^4=1$.
\\
\textup{(3)}  
For $\xi\in \Z_2^n$, there exists $\theta\in \mathrm{Aut}_F(V)$ such that 
$\theta|_{V^0}=\sigma_\xi$ if and only if $\xi \cdot D\subset C$, and in this case, 
the order $|\theta|=2$ if and only if $\xi \cdot D$ is a doubly even subcode of $C$, 
and otherwise $|\theta|=4$.
\\
\textup{(4)} 
Let $P=\{ \xi\in \Z_2^n \mid \xi \cdot D\subset C\}$.
Then $C^\perp \subset P$ and we have an exact sequence
\[
  1\longrightarrow \Z_2^n/D^\perp \longrightarrow \mathrm{Aut}_F(V) \longrightarrow P/C^\perp
  \longrightarrow 1.
\]
\end{theorem}

We review a construction of twisted modules from \cite{LamYamauchi08}.
Suppose we have a codeword $\xi\in P$ such that $\xi\cdot D$ is a doubly even 
subcode of $C$.
Let $\theta\in \mathrm{Aut}_F(V)$ be an involution such that $\theta|_{V^0}=\sigma_{\xi}$.
Note that such a $\theta_\xi$ is not unique and only determined modulo $\mathrm{Im}\,\tau$ by 
(4) of Theorem \ref{thm:2}.
The fixed point subalgebra $V^+=V^{\langle \theta\rangle}$ is a framed VOA with structure codes 
$(C^0,D)$ where $C^0=C\cap \langle \xi\rangle^\perp$.
We denote its $1/16$-word decomposition by $V^+=\oplus_{\alpha\in D}V^{+,\alpha}$.
Let $X$ be an irreducible $F$-module isomorphic to 
$L(\sfr{1}{2},\sfr{1}{16})^{\otimes \mathrm{wt}(\xi)}\otimes L(\sfr{1}{2},0)^{\otimes n-\mathrm{wt}(\xi)}$
whose $1/16$-word is $\xi$.
There exists an irreducible $V^{+,0}$-module $Y$ which contains $X$ as an $F$-submodule 
(cf.~\cite{LamYamauchi08,Miyamoto98}).
Since $V=\oplus_{\alpha \in D}V^\alpha$ is a $\Z_2\oplus D$-graded simple current extension of 
$V^{+,0}$, there exists $\tau_\eta\in \mathrm{Im}\,\tau$ such that 
the fusion product
\begin{equation}\label{eq:4}
  W = V\boxtimes_{V^{+,0}} Y
  =\bigoplus_{\alpha\in D} V^\alpha \boxtimes_{V^{+,0}} Y
\end{equation}
has a unique structure of an irreducible $\theta\tau_\eta$-twisted $V$-module.
Since $\theta$ and $\theta\tau_\eta$ define the same automorphism $\sigma_\xi$ on the subalgebra 
$V^0$ of $V$, by replacing $\theta$ by $\theta\tau_\eta$ if necessary, 
we may regard $W$ as an irreducible $\theta$-twisted $V$-module. 
Each summand $V^\alpha \boxtimes_{V^{+,0}} Y$, $\alpha \in D$, of $W$ has the 1/16-word 
$\alpha+\xi$ so that its top weight is at least $\mathrm{wt}(\xi+\alpha)/16$.
Summarizing, we have the following.

\begin{theorem}\label{thm:3}
Let $V=\oplus_{\alpha \in D}V^\alpha$ be a framed VOA with structure codes $(C,D)$ 
with respect to a frame $F\cong L(\sfr{1}{2},0)^{\otimes n}$.
Let $\xi\in P$ be a codeword such that $\xi\cdot D$ is a doubly even subcode of $C$.
Let $X$ be the irreducible $F$-module isomorphic to 
$L(\sfr{1}{2},\sfr{1}{16})^{\otimes \mathrm{wt}(\xi)}\otimes L(\sfr{1}{2},0)^{\otimes n-\mathrm{wt}(\xi)}$ 
such that its 1/16-word is $\xi$.
Then there exists an involutive automorphism $\theta\in \mathrm{Aut}(V)$ and 
an irreducible $\theta$-twisted $V$-module $W$ such that 
$\theta|_{V^0}=\sigma_\xi$ and $W$ contains $X$ as an irreducible $F$-submodule.
The top weight of $W$ is at least $\min\{\mathrm{wt}(\xi+\alpha)/16\mid \alpha \in D\}$.
\end{theorem}

Recall the Reed-Muller codes.
The first order Reed-Muller code $\mathrm{RM}(1,4)$ of length $2^4$ 
is defined by the following generating matrix.
\begin{equation}\label{eq:5}
\begin{bmatrix}
  1111 & 1111 & 1111 & 1111
  \\ 
  1111 & 1111 & 0000 & 0000 
  \\
  1111 & 0000 & 1111 & 0000 
  \\ 
  1100 & 1100 & 1100 & 1100 
  \\ 
  1010 & 1010 & 1010 & 1010
\end{bmatrix}
\end{equation}
The dual code of $\mathrm{RM}(1,4)$ is the second order Reed-Muller code 
$\mathrm{RM}(2,4)=\mathrm{RM}(1,4) \cdot \mathrm{RM}(1,4)$.
The first order Reed-Muller code $\mathrm{RM}(1,6)$ of length $2^6$ is defined by
\begin{equation}\label{eq:6}
  \mathrm{RM}(1,6)
  = \mathrm{Span}_{\Z_2}\{ (\alpha,\alpha,\alpha,\alpha),~(0^{16}1^{16}0^{16}1^{16}),~(0^{32}1^{32})
  \mid \alpha \in \mathrm{RM}(1,4)\} .
\end{equation}
The weight enumerator of $\mathrm{RM}(1,6)$ is $x^{64}+126x^{32}+1$.
The dual code of $\mathrm{RM}(1,6)$ is the 4th order Reed-Muller code $\mathrm{RM}(4,6)$ 
of length $2^6$.
It follows from the MacWilliams identity that the minimum weight of $\mathrm{RM}(4,6)$ is 4.
It is easy to see that
\begin{equation}\label{eq:7}
  \begin{matrix}
  (\alpha,\beta,\gamma,\delta)\in \mathrm{RM}(4,6) 
  \medskip\\
  (\alpha, \beta, \gamma, \delta\in \Z_2^{16})
  \end{matrix}
  ~\iff~
  \begin{matrix}
    \alpha+\beta+\gamma+\delta\in \mathrm{RM}(2,4)=\mathrm{RM}(1,4)^\perp,
    \medskip\\
    \mathrm{wt}(\alpha) \equiv \mathrm{wt}(\beta)\equiv \mathrm{wt}(\gamma) \equiv \mathrm{wt}(\delta) \bmod 2.
  \end{matrix}
\end{equation}

Since $\mathrm{RM}(1,6)$ is triply even, there exists a $c=32$ holomorphic framed VOA with 
structure codes $(\mathrm{RM}(4,6),\mathrm{RM}(1,6))$ by Remark 6 and Theorem 10 of \cite{LamYamauchi08}.
Since the minimum weights of $\mathrm{RM}(1,6)$ and $\mathrm{RM}(4,6)$ are 32 and 4, 
respectively, the weight one subspace of such a VOA is trivial.
It follows from Theorem 3.13 of \cite{LamShimakura2015} that such a framed VOA is uniquely determined 
by its structure codes.
On the other hand, it is shown in \cite{Miyamoto96Automorphism} that the $\Z_2$-orbifold construction 
$\tilde{V}_{\mathrm{BW}_{32}}$ of the lattice VOA associated with the Barnes-Wall lattice 
$\mathrm{BW}_{32}$ of rank 32 is a $c=32$ holomorphic framed VOA with structure codes 
$(\mathrm{RM}(4,6),\mathrm{RM}(1,6))$.
The VOA $\tilde{V}_{\mathrm{BW}_{32}}$ is studied in \cite{Miyamoto96Automorphism,Shimakura14} and it has a finite 
automorphism group of the shape ${2^{27}}.E_6(2)$.

\begin{theorem}\label{thm:4}
  The $\Z_2$-orbifolding $\tilde{V}_{\mathrm{BW}_{32}}$ is the unique $c=32$ holomorphic 
  framed VOA with structure codes $(\mathrm{RM}(4,6),\mathrm{RM}(1,6))$.
\end{theorem}

We will prove Theorem \ref{thm:1} by using $\tilde{V}_{\mathrm{BW}_{32}}$.
Let $\alpha^1=(1^{16})$, $\alpha^2=(1^80^8)$, $\alpha^3=(1^40^41^40^4)$, 
$\alpha^4=(\{1^20^2\}^4)$, $\alpha^5=(\{10\}^8)$ be the basis of $\mathrm{RM}(1,4)$ in 
\eqref{eq:5}.
Then by \eqref{eq:6} the Reed-Muller code $\mathrm{RM}(1,6)$ has a basis
\begin{equation}\label{eq:8}
  \gamma^i=(\alpha^i,\alpha^i,\alpha^i,\alpha^i),~~~1\leq i\leq 5,~~~
  \gamma^6=(0^{16}1^{16}0^{16}1^{16}),~~~
  \gamma^7=(0^{32}1^{32}).
\end{equation}

\begin{lem}\label{lem:5}
  Let $\nu^1$, $\nu^2$, $\nu^3$, $\nu^4\in \Z_2^{16}$ and let
  $\xi=(\nu^1,\nu^2,\nu^3,\nu^4)\in \Z_2^{64}$.
  Then $\xi\cdot \mathrm{RM}(1,6)$ is a subcode of $\mathrm{RM}(4,6)$ if and only if 
  the following conditions are satisfied.
  \\
  \textup{(i)} $\nu^1+\nu^2+\nu^3+\nu^4\in \mathrm{RM}(1,4)$, 
  \\ 
  \textup{(ii)} $\nu^i+\nu^j\in \mathrm{RM}(1,4)^\perp$ for $1\leq i<j\leq 4$,
  \\
  \textup{(iii)} $\nu^1$, $\nu^2$, $\nu^3$, $\nu^4$ are even. 
  \\
  Moreover, $\xi\cdot \mathrm{RM}(1,6)$ is a doubly even subcode of $\mathrm{RM}(4,6)$ 
  if and only if it further satisfies the following condition.
  \\
  \textup{(iv)} $\xi\cdot \gamma^i$, $1\leq i\leq 5$, are doubly even.
\end{lem}

\proof
First, we prove that $\xi\cdot \mathrm{RM}(1,6)$ is a subcode of 
$\mathrm{RM}(4,6)=\mathrm{RM}(1,6)^\perp$ if and only if $\xi$ satisfies 
the conditions (i), (ii) and (iii).
Let $\alpha$, $\beta\in \mathrm{RM}(1,6)$.
We have $(\xi\cdot \alpha |\beta)=(\xi |\alpha\cdot \beta)$ 
so that $\xi\cdot \mathrm{RM}(1,6)$ is a subcode of $\mathrm{RM}(1,6)^\perp$ if and only if 
$(\xi| \gamma^i\cdot \gamma^j)=0$ for $1\leq i\leq j\leq 7$.
For $1\leq i,j\leq 5$, we have 
\[
\begin{array}{ll}
  (\xi| \gamma^i\cdot \gamma^j)
  &=((\nu^1,\nu^2,\nu^3,\nu^4)|(\alpha^i\cdot \alpha^j,\alpha^i\cdot \alpha^j,
  \alpha^i\cdot \alpha^j, \alpha^i\cdot \alpha^j)
  \medskip\\
  &= (\nu^1+\nu^2+\nu^3+\nu^4| \alpha^i\cdot \alpha^j)
\end{array}
\]
so that $(\xi| \gamma^i\cdot \gamma^j)=0$ for $1\leq i,j\leq 5$ if and only if 
$\nu^1+\nu^2+\nu^3+\nu^4\in (\mathrm{RM}(1,4)\cdot \mathrm{RM}(1,4))^\perp
=\mathrm{RM}(2,4)^\perp=\mathrm{RM}(1,4)$.
Thus we obtain the condition (i).
Similarly, from $(\xi| \gamma^i\cdot \gamma^j)=0$ for $1\leq i\leq 5$ and $j=6$, $7$ 
we obtain $\nu^2+\nu^4$, $\nu^3+\nu^4\in \mathrm{RM}(1,4)^\perp$.
Since $\mathrm{RM}(1,4)\subset \mathrm{RM}(2,4)$, it follows that 
$\nu^i+\nu^j\in \mathrm{RM}(1,4)^\perp$ for $1\leq i<j\leq 4$ and we obtain the condition (ii).
Also, from $(\xi| \gamma^i\cdot \gamma^j)=0$ for $6\leq i,j\leq 7$ we obtain the condition (iii).
Thus $\xi\cdot \mathrm{RM}(1,6)$ is a subcode of $\mathrm{RM}(4,6)$ if and only if 
$\xi$ satisfies the conditions (i)--(iii).

Now suppose $\xi\cdot \mathrm{RM}(1,6)$ is a subcode of $\mathrm{RM}(4,6)$.
As we have discussed, this is equivalent to that 
$(\xi\cdot \gamma^i | \xi\cdot \gamma^j)=(\xi| \gamma^i\cdot \gamma^j)=0$ for $1\leq i,j\leq 7$ 
so that $\xi\cdot \mathrm{RM}(1,6)$ is self-orthogonal.
Since a sum of mutually orthogonal doubly even codewords is again doubly even, 
$\xi\cdot \mathrm{RM}(1,6)$ is doubly even if and only if all seven vectors $\xi\cdot \gamma^i$, 
$1\leq i\leq 7$, are doubly even.
It follows from the condition (iii) that $\xi\cdot \gamma^6$ and $\xi\cdot\gamma^7$ are doubly even.
Therefore, $\xi\cdot \mathrm{RM}(1,6)$ is doubly even if and only if $\xi$ satisfies 
the condition (iv).
\qed

\begin{lem}\label{lem:6}
Let $\alpha$ be a weight 6 codeword of $\mathrm{RM}(2,4)$. 
Then the codeword 
\[
  \xi=(\alpha,\alpha,\alpha,\alpha^c)\in \Z_2^{64}
\]
satisfies the conditions \textup{(i)--(iv)} in Lemma \ref{lem:5}, 
where $\alpha^c=\mathbf{1}_{16}+\alpha$.
The weight enumerator of the coset $\xi+\mathrm{RM}(1,6)$ is 
$64x^{28}+64x^{36}$.
\end{lem}

\paragraph{\textbf{Proof of Theorem \ref{thm:1}:}}

Let $V=\oplus_{\alpha\in \mathrm{RM}(1,6)}V^\alpha$ be a holomorphic framed VOA with structure codes 
$(\mathrm{RM}(4,6),\mathrm{RM}(1,6))$.
Then $V\cong \tilde{V}_{\mathrm{BW}_{32}}$ by Theorem \ref{thm:4}.
Let 
$$
\alpha=(0110~1100~1010~0000)\in \mathrm{RM}(2,4)
$$
and set
\begin{equation}\label{eq:3.5}
  \xi=(\alpha,\alpha,\alpha,\alpha^c)\in \Z_2^{64}.
\end{equation}
Let $X$ be an irreducible $L(\sfr{1}{2},0)^{\otimes 64}$-module 
isomorphic to $L(\sfr{1}{2},\sfr{1}{16})^{\otimes 28}\otimes L(\sfr{1}{2},0)^{\otimes 36}$ such that
its $1/16$-word is $\xi$.
By Theorem \ref{thm:3} and Lemmas \ref{lem:5}, \ref{lem:6}, there is an involution 
$\theta\in \mathrm{Aut}(V)$ such that $\theta|_{V^0}=\sigma_\xi$ and 
the irreducible $\theta$-twisted $V$-module $W$ has top weight $28/16=7/4$.
\qed

\baselineskip 14pt

\newpage

\appendix

\section{Data}\label{sec: A data}

As in \cite{TenerWang17} we can compute potential characters vectors for each of the allowed genera of Theorem \ref{thm: main theorem}.
The realizations labelled \cite{GHM} arise as cosets of affine VOAs.
{\small
\begin{tablenofloat}
$$
\renewcommand*{\arraystretch}{1.1}
\begin{array}{|c|c|c|c|}
\hline
\cC & c & \mbox{Realization} & \mbox{Character}\\
\hline
\Semion & 1 & A_{1,1} &
q^{-1/24}\left(
\begin{array}{l}
1 + 3q + 4q^2 + \cdots\\
q^{\frac{1}{4}}(2 + 2q + 6q^2 + \cdots)
\end{array}
\right)\\
\hline
\Semion & 9 & A_{1,1} \otimes E_{8,1} &
q^{-9/24}\left(
\begin{array}{l}
1 + 251q + 4872q^2 + \cdots\\
q^{\frac{1}{4}}(2 + 498q + 8750q^2 + \cdots)
\end{array}
\right)\\
\hline
\Semion & 17 & \mbox{\cite{GHM}} &
q^{-17/24}\left(
\begin{array}{l}
1 + 323q + 60860q^2 + \cdots\\
q^{\frac{5}{4}}(1632 + 162656q + 4681120q^2 + \cdots)
\end{array}
\right)\\
\hline
\Semion & 33 & \S \ref{sec: new VOA} &
q^{-33/24}\left(
\begin{array}{l}
1 + 3q + 86004q^2 + \cdots\\
q^{\frac{9}{4}}(565760 + 192053760q + \cdots)
\end{array}
\right)\\
\hline
\overline{\Semion} & 7 & E_{7,1} &
q^{-7/24}\left(
\begin{array}{l}
1 + 133q + 1673q^2 + \cdots\\
q^{\frac{3}{4}}(56 + 968q + 7504q^2 + \cdots)
\end{array}
\right)\\
\hline
\overline{\Semion} & 15 & E_{7,1} \otimes E_{8,1} &
q^{-15/24}\left(
\begin{array}{l}
1 + 381q + 38781q^2 + \cdots\\
q^{\frac{3}{4}}(56 + 14856q + 478512q^2 + \cdots)
\end{array}
\right)\\
\hline
\overline{\Semion} & 23 & \mbox{\cite{GHM}} &
q^{-23/24}\left(
\begin{array}{l}
1 + 69q + 131905q^2 + \cdots\\
q^{\frac{7}{4}}(32384 + 4418944q + 189846784q^2 + \cdots)
\end{array}
\right)\\
\hline
\Fib & \frac{14}{5} & G_{2,1} &
q^{-7/60}\left(
\begin{array}{l}
1 + 14q + 42q^2 + \cdots\\
q^{\frac{2}{5}}(7 + 34q + 119q^2 + \cdots)
\end{array}
\right)\\
\hline
\Fib & \frac{54}{5} & G_{2,1} \otimes E_{8,1} &
q^{-27/60}\left(
\begin{array}{l}
1 + 262q + 7638q^2 + \cdots\\
q^{\frac{2}{5}}(7 + 1770q + 37419q^2 + \cdots)
\end{array}
\right)\\
\hline
\Fib & \frac{94}{5} & \mbox{\cite{GHM}} &
q^{-47/60}\left(
\begin{array}{l}
1 + 188q + 62087q^2 + \cdots\\
q^{\frac{7}{5}}(4794 + 532134q + 17518686q^2 + \cdots)
\end{array}
\right)\\
\hline
\overline{\Fib} & \frac{26}{5} & F_{4,1} &
q^{-13/60}\left(
\begin{array}{l}
1 + 52q + 377q^2 + \cdots\\
q^{\frac{3}{5}}(26 + 299q + 1702q^2 + \cdots)
\end{array}
\right)\\
\hline
\overline{\Fib} & \frac{66}{5} & F_{4,1} \otimes E_{8,1} &
q^{-33/60}\left(
\begin{array}{l}
1 + 300q + 17397q^2 + \cdots\\
q^{\frac{3}{5}}(26 + 6747q + 183078q^2 + \cdots)
\end{array}
\right)\\
\hline
\overline{\Fib} & \frac{106}{5} & \mbox{\cite{GHM}} &
q^{-33/60}\left(
\begin{array}{l}
1 + 106q + 84429q^2 +\cdots\\
q^{\frac{8}{5}}(15847 + 1991846q + 76895739q^2 + \cdots)
\end{array}
\right)\\
\hline
\LeeYang & -\frac{22}{5} & \LeeYang &
q^{11/60}\left(
\begin{array}{l}
1 + 0q + q^2 +\cdots\\
q^{-\frac{1}{5}}(1 + q + q^2 + \cdots)
\end{array}
\right)\\ 
\hline
\LeeYang & \frac{18}{5} & \operatorname{Y-L} \otimes E_{8,1}&
q^{-3/20}\left(
\begin{array}{l}
1 + 248q + 4125q^2 +\cdots\\
q^{-\frac{1}{5}}(1 + 249q + 4373q^2 + \cdots)
\end{array}
\right)\\ 
\hline
\end{array}
$$
\smallskip
\caption{Characters of strongly rational extremal VOAs with two simple modules}
\label{tab: characters}
\end{tablenofloat}
}

\newpage

The following table (which was used in the proof of Theorem \ref{thm: cmax}) has 24 rows, in groups of 3.
Each group corresponds to a modular tensor category from Table \ref{tab: 8 cases}, and each row within the group selects a representative of an equivalence class of admissible $c$ mod $24$.
The table provides the characteristic matrix (computed as in \cite{TenerWang17}), the extremal conformal dimension $h_{ext}(c)$ computed from the definition \eqref{eqn: extremality}, and a value $n_{max}$ such that $\chi(c+24n)_{00} < 0$ when $n > n_{max}$, as computed using Lemma \ref{lem: positive c recurrence}.

{\small
\renewcommand*{\arraystretch}{1.1}
\begin{longtable}{|c|c|c|c|c|c|}
\hline
$\#$ & $\cC$ & $c$ & $\chi(c)$ & $h_{ext}(c)$ & $n_{max}$ \endhead
\hline
$1$  & $\,\,\Semion \,\,$ & $\,\, 1 \,\,$ & $\begin{pmatrix} 3 & 26752\\ 2 & -247\end{pmatrix}$ & $\frac14$ & $0$\\
\hline
$1$  & $\Semion$  & $9$ & $\begin{pmatrix} 251 & 26752\\ 2 & 1\end{pmatrix}$ & $\frac14$ & $2$ \\
\hline
$1$  & $\Semion$ & $17$ & $\begin{pmatrix} 323 & 88\\ 1632 & -319\end{pmatrix}$ & $\frac54$ & $0$ \\
\hline
$2$ & $\overline{\Semion}$ & $7$ & $\begin{pmatrix} 133 & 1248\\ 56 & -377\end{pmatrix}$ & $\frac34$ & $0$ \\
\hline
$2$ & $\overline{\Semion}$ & $15$ & $\begin{pmatrix} 381 & 1248\\ 56 & -129\end{pmatrix}$ & $\frac34$ & $1$ \\
\hline
$2$ & $\overline{\Semion}$ & $23$ & $\begin{pmatrix} 69 & 10\\ 32384 & -65\end{pmatrix}$ & $\frac74$ & $0$ \\
\hline
$3$ & $\Semion^\dagger$ & $11$ & $\begin{pmatrix} -319 & 1632\\ 88 & 323\end{pmatrix}$ & $\frac34$ & $0$ \\
\hline
$3$ & $\Semion^\dagger$ & $19$ & $\begin{pmatrix} -247 & 2\\ 26752 & 3\end{pmatrix}$ & $\frac74$ & $2$ \\
\hline
$3$ & $\Semion^\dagger$ & $27$ & $\begin{pmatrix} 1 & 2\\ 26752 & 251\end{pmatrix}$ & $\frac74$ & $0$ \\
\hline
$4$ & $\overline{\Semion}^\dagger$ & $5$ & $\begin{pmatrix} -65 & 32384\\ 10 & 69\end{pmatrix}$ & $\frac14$ & $0$ \\
\hline
$4$ & $\overline{\Semion}^\dagger$ & $13$ & $\begin{pmatrix} -377 & 56\\ 1248 & 133\end{pmatrix}$ & $\frac54$ & $1$ \\
\hline
$4$ & $\overline{\Semion}^\dagger$ & $21$ & $\begin{pmatrix} -129 & 56\\ 1248 & 381\end{pmatrix}$ & $\frac54$ & $0$ \\
\hline
$5$ & $\Fib$ & $\frac{14}5$ & $\begin{pmatrix} 14 & 12857\\ 7 & -258\end{pmatrix}$ & $\frac25$ & $0$ \\
\hline
$5$ & $\Fib$ & $\frac{54}5$ & $\begin{pmatrix} 262 & 12857\\ 7 & -10\end{pmatrix}$ & $\frac25$ & $1$ \\
\hline
$5$ & $\Fib$ & $\frac{94}5$ & $\begin{pmatrix} 188 & 46\\ 4794 & -184\end{pmatrix}$ & $\frac75$ & $0$ \\
\hline
$6$ & $\overline{\Fib}$ & $\frac{26}5$ & $\begin{pmatrix} 52 & 3774\\ 26 & -296\end{pmatrix}$ & $\frac35$ & $0$ \\
\hline
$6$ & $\overline{\Fib}$ & $\frac{66}5$ & $\begin{pmatrix} 300 & 3774\\ 26 & -48\end{pmatrix}$ & $\frac35$ & $1$ \\
\hline
$6$ & $\overline{\Fib}$ & $\frac{106}5$ & $\begin{pmatrix} 106 & 17\\ 15847 & -102\end{pmatrix}$ & $\frac85$ & $0$ \\
\hline
$7$ & $\LeeYang$ & $\frac{58}5$ & $\begin{pmatrix} -406 & 902\\ 87 & 410\end{pmatrix}$ & $\frac45$ & $0$ \\
\hline
$7$ & $\LeeYang$ & $\frac{98}5$ & $\begin{pmatrix} -245 & 1\\ 26999 & 1\end{pmatrix}$ & $\frac95$ & $2$ \\
\hline
$7$ & $\LeeYang$ & $\frac{138}5$ & $\begin{pmatrix} 3 & 1\\ 26999 & 249\end{pmatrix}$ & $\frac95$ & $0$ \\
\hline
$8$ & $\overline{\LeeYang}$ & $\frac{22}5$ & $\begin{pmatrix} -55 & 32509\\ 11 & 59\end{pmatrix}$ & $\frac15$ & $1$ \\
\hline
$8$ & $\overline{\LeeYang}$ & $\frac{62}5$ & $\begin{pmatrix} -434 & 57\\ 682 & 190\end{pmatrix}$ & $\frac65$ & $1$ \\
\hline
$8$ & $\overline{\LeeYang}$ & $\frac{102}5$ & $\begin{pmatrix} -186 & 57\\ 682 & 438\end{pmatrix}$ & $\frac65$ & $1$\\
\hline
\caption{Values of $n_{max}$ computed from Lemma \ref{lem: positive c recurrence}}
\label{tab: nmax}
\end{longtable}
}

The following table (which was used in the proof of Theorem \ref{thm: cmin}) again has 24 rows, in the same groups of 3.
Each group corresponds to a modular tensor category from Table \ref{tab: 8 cases}, and each row within the group selects a representative of an equivalence class of admissible $c$ mod $24$.
The table provides the characteristic matrix (computed as in \cite{TenerWang17}), the extremal conformal dimension $h_{ext}(c)$ computed from the definition \eqref{eqn: extremality}, $\beta(c)$ and $\alpha(c)$ computed directly from the characteristic matrix.
Using this data, we apply Lemma \ref{lem: negative c bound} to obtain a value $n_{max}$ such that $\abs{\beta(c-24n)} > 1$ when $n > n_{max}$.
In fact, $n_{max} = 0$ in all cases.

{\small
\renewcommand*{\arraystretch}{1.35}
\begin{longtable}{|c|c|c|c|c|c|c|c|c|}
\hline
$\#$ & $\cC$ & $c$ & $\chi(c)$ & $h_{ext}(c)$ & $\alpha(c)$ & $\beta(c)$ & $n_{max}$ & $\chi_{10}(c)$ \endhead
\hline
$1$ & $\,\,\Semion\,\,$ & $\,\,-7\,\,$ & $\begin{pmatrix} 59 & 13424640\\ \frac{1}{88} & -55 \end{pmatrix}$ & $-\frac34$ & $114$ & $\frac{1678080}{11}$ & $0$ & $\frac1{88}$ \\
\hline
$1$ & $\Semion$  & $-15$ & $\begin{pmatrix} \frac{3441}{11} & \frac{57264144384}{11}\\ \frac{1}{26752} & -\frac{669}{11} \end{pmatrix}$ & $-\frac74$ & $\frac{4110}{11}$ & $\frac{23546112}{121}$ & $0$ & $\frac1{26752}$ \\
\hline
$1$ & $\Semion$ & $-23$ & $\begin{pmatrix} \frac{713}{11} & \frac{57264144384}{11}\\ \frac{1}{26752} & -\frac{3397}{11} \end{pmatrix}$ & $-\frac74$ & $\frac{4110}{11}$ & $\frac{23546112}{121}$ & $0$ & $\frac1{26752}$ \\
\hline
$2$ & $\overline{\Semion}$ & $-1$ & $\begin{pmatrix} \frac{49}5 & \frac{3281408}5\\ \frac1{10} & \frac{-29}5\end{pmatrix}$ & $\frac{-1}4$ & $\frac{78}5$ & $\frac{1640704}{25}$ & $0$ & $\frac1{10}$ \\
\hline
$2$ & $\overline{\Semion}$ & $-9$ & $\begin{pmatrix} \frac{863}3 & \frac{747151360}3\\ \frac1{1248} & \frac{-107}3\end{pmatrix}$ & $\frac{-5}4$ & $\frac{970}3$ & $\frac{23348480}{117}$ & $0$ & $\frac1{1248}$ \\
\hline
$2$ & $\overline{\Semion}$ & $-17$ & $\begin{pmatrix} \frac{119}3 & \frac{747151360}3\\ \frac1{1248} & \frac{-851}3\end{pmatrix}$ & $\frac{-5}4$ & $\frac{970}3$ & $\frac{23348480}{117}$ & $0$ & $\frac1{1248}$ \\
\hline
$3$ & $\Semion^\dagger$ & $3$ & $\begin{pmatrix} 249 & 565760\\ \frac12 & 3\end{pmatrix}$ & $\frac{-1}4$ & $246$ & $282880$ & $0$ & $\frac12$ \\
\hline
$3$ & $\Semion^\dagger$ & $-5$ & $\begin{pmatrix} 1 & 565760\\ \frac12 & -245\end{pmatrix}$ & $\frac{-1}4$ & $246$ & $282880$ & $0$ & $\frac12$ \\
\hline
$3$ & $\Semion^\dagger$ & $-13$ & $\begin{pmatrix} \frac{299}3 & \frac{827924480}3\\ \frac1{1632} & \frac{-287}3\end{pmatrix}$ & $\frac{-5}4$ & $\frac{586}3$ & $\frac{1521920}9$ & $0$ & $\frac1{1632}$ \\
\hline
$4$ & $\overline{\Semion}^\dagger$ & $-3$ & $\begin{pmatrix} \frac{1857}7 & \frac{83232768}7\\ \frac1{56} & \frac{-93}7\end{pmatrix}$ & $\frac{-3}4$ & $\frac{1950}7$ & $\frac{10404096}{49}$ & $0$ & $\frac1{56}$ \\
\hline
$4$ & $\overline{\Semion}^\dagger$ & $-11$ & $\begin{pmatrix} \frac{121}7 & \frac{827924480}3\\ \frac1{1632} & \frac{-1829}7\end{pmatrix}$ & $\frac{-3}4$ & $\frac{1950}7$ & $\frac{10404096}{49}$ & $0$ & $\frac1{56}$ \\
\hline
$4$ & $\overline{\Semion}^\dagger$ & $-19$ & $\begin{pmatrix} \frac{1501}{11} & \frac{62591041536}{11}\\ \frac1{32384} & \frac{-1457}{11}\end{pmatrix}$ & $\frac{-7}4$ & $\frac{2958}{11}$ & $\frac{21260544}{121}$ & $0$ & $\frac1{32384}$ \\
\hline
$5$ & $\Fib$ & $\frac{-26}5$ & $\begin{pmatrix} \frac{91}2 & \frac{13051833}2\\ \frac1{46} & \frac{-83}2\end{pmatrix}$ & $\frac{-8}5$ & $87$ & $\frac{567471}4$ & $0$ & $\frac1{46}$ \\
\hline
$5$ & $\Fib$ & $\frac{-66}5$ & $\begin{pmatrix} \frac{3966}{13} & \frac{32712244109}{13}\\ \frac1{12857} & \frac{-690}{13}\end{pmatrix}$ & $\frac{-8}5$ & $\frac{4656}{13}$ & $\frac{33076081}{169}$ & $0$ & $\frac1{12857}$ \\
\hline
$5$ & $\Fib$ & $\frac{-106}5$ & $\begin{pmatrix} \frac{742}{13} & \frac{32712244109}{13}\\ \frac1{12857} & \frac{-3914}{13}\end{pmatrix}$ & $\frac{-8}5$ & $\frac{4656}{13}$ & $\frac{33076081}{169}$ & $0$ & $\frac1{12857}$ \\
\hline
$6$ & $\overline{\Fib}$ & $\frac{-14}5$ & $\begin{pmatrix} 26 & 1951158\\ \frac1{17} & -22\end{pmatrix}$ & $\frac{-2}5$ & $48$ & $114774$ & $0$ & $\frac1{17}$ \\
\hline
$6$ & $\overline{\Fib}$ & $\frac{-54}5$ & $\begin{pmatrix} 295 & 745916226\\ \frac1{3774} & -43\end{pmatrix}$ & $\frac{-7}5$ & $338$ & $\frac{3359983}{17}$ & $0$ & $\frac1{3774}$ \\
\hline
$6$ & $\overline{\Fib}$ & $\frac{-94}5$ & $\begin{pmatrix} 47 & 745916226\\ \frac1{3774} & -291\end{pmatrix}$ & $\frac{-7}5$ & $338$ & $\frac{3359983}{17}$ & $0$ & $\frac1{3774}$ \\
\hline
$7$ & $\LeeYang$ & $\frac{18}5$ & $\begin{pmatrix} 248 & 310124\\ 1 & 4\end{pmatrix}$ & $\frac{-1}5$ & $244$ & $310124$ & $0$ & $1$ \\
\hline
$7$ & $\LeeYang$ & $\frac{-22}5$ & $\begin{pmatrix} 0 & 310124\\ 1 & -244\end{pmatrix}$ & $\frac{-1}5$ & $244$ & $310124$ & $0$ & $1$ \\
\hline
$7$ & $\LeeYang$ & $\frac{-62}5$ & $\begin{pmatrix} \frac{1054}{11} & \frac{1667924403}{11}\\ \frac1{902} & \frac{-1010}{11}\end{pmatrix}$ & $\frac{-6}5$ & $\frac{2064}{11}$ & $\frac{40681083}{242}$ & $0$ & $\frac1{902}$ \\
\hline
$8$ & $\overline{\LeeYang}$ & $\frac{-18}5$ & $\begin{pmatrix} \frac{802}3 & \frac{35954954}3\\ \frac1{57} & \frac{-46}3\end{pmatrix}$ & $\frac{-4}5$ & $\frac{848}3$ & $\frac{1892366}9$ & $0$ & $\frac1{57}$ \\
\hline
$8$ & $\overline{\LeeYang}$ & $\frac{-58}5$ & $\begin{pmatrix} \frac{58}3 & \frac{35954954}3\\ \frac1{57} & \frac{-790}3\end{pmatrix}$ & $\frac{-4}5$ & $\frac{848}3$ & $\frac{1892366}9$ & $0$ & $\frac1{57}$ \\
\hline
$8$ & $\overline{\LeeYang}$ & $\frac{-98}5$ & $\begin{pmatrix} 140 & 5726299516\\ \frac1{323509} & -136\end{pmatrix}$ & $\frac{-9}5$ & $276$ & $\frac{3346756}{19}$ & $0$ & $\frac1{323509}$ \\
\hline

\caption{Values of $n_{max}$ computed from Lemma \ref{lem: negative c bound}}
\label{tab: nmax negative}
\end{longtable}
\renewcommand*{\arraystretch}{1}
}

\newpage


\begin{thebibliography}{EGNO15}

\bibitem[ABD04]{ABD04}
T.~Abe, G.~Buhl, and C.~Dong.
\newblock Rationality, regularity, and {$C_2$}-cofiniteness.
\newblock {\em Trans. Amer. Math. Soc.}, 356(8):3391--3402, 2004.

\bibitem[BG07]{BG}
P.~Bantay and T.~Gannon.
\newblock Vector-valued modular functions for the modular group and the
  hypergeometric equation.
\newblock {\em Commun. Number Theory Phys.}, 1(4):651--680, 2007.

\bibitem[CM16]{CarnahanMiyamoto16}
S.~Carnahan and M.~Miyamoto.
\newblock Regularity of fixed-point vertex operator subalgebras.
\newblock {\em arXiv:1603.05645 [math.RT]}, 2016.

\bibitem[CM19]{ChandraMukhi2019}
A.~R. Chandra and S.~Mukhi.
\newblock Towards a classification of two-character rational conformal field
  theories.
\newblock {\em Journal of High Energy Physics}, 2019(4):153, 2019.

\bibitem[DGH98]{DongGriessHoehn98}
C.~Dong, R.~L. Griess, Jr., and G.~H\"{o}hn.
\newblock Framed vertex operator algebras, codes and the {M}oonshine module.
\newblock {\em Comm. Math. Phys.}, 193(2):407--448, 1998.

\bibitem[DLM97]{DLM97}
C.~Dong, H.~Li, and G.~Mason.
\newblock Regularity of rational vertex operator algebras.
\newblock {\em Adv. Math.}, 132(1):148--166, 1997.

\bibitem[DLN15]{DongLinNg15}
C.~Dong, X.~Lin, and S.-H. Ng.
\newblock Congruence property in conformal field theory.
\newblock {\em Algebra Number Theory}, 9(9):2121--2166, 2015.

\bibitem[EG11]{EvansGannon11}
D.~E. Evans and T.~Gannon.
\newblock The exoticness and realisability of twisted {H}aagerup-{I}zumi
  modular data.
\newblock {\em Comm. Math. Phys.}, 307(2):463--512, 2011.

\bibitem[EGNO15]{EGNOBook}
P.~Etingof, S.~Gelaki, D.~Nikshych, and V.~Ostrik.
\newblock {\em Tensor categories}, volume 205 of {\em Mathematical Surveys and
  Monographs}.
\newblock American Mathematical Society, Providence, RI, 2015.

\bibitem[FM14]{FrancMason14}
C.~Franc and G.~Mason.
\newblock Fourier coefficients of vector-valued modular forms of dimension 2.
\newblock {\em Canad. Math. Bull.}, 57(3):485--494, 2014.

\bibitem[Gan14]{Gannon14}
T.~Gannon.
\newblock The theory of vector-valued modular forms for the modular group.
\newblock In {\em {C}onformal {F}ield {T}heory, {A}utomorphic {F}orms and
  {R}elated {T}opics}, pages 247--286. Springer, 2014.

\bibitem[GHM16]{GHM}
M.~R. Gaberdiel, H.~R. Hampapura, and S.~Mukhi.
\newblock Cosets of meromorphic {CFT}s and modular differential equations.
\newblock {\em J. High Energy Phys.}, 2016(4):1--13, 2016.

\bibitem[Gra18]{GradyThesis}
J.~C. Grady.
\newblock The classification of extremal vertex operator algebras of rank 2.
\newblock Undergraduate thesis, UC Santa Barbara, July 2018.
\newblock Available online at \url{http://math.tener.cc/}.

\bibitem[H{\"o}h95]{Hoehn95}
G.~H{\"o}hn.
\newblock {\em Selbstduale Vertexoperatorsuperalgebren und das Babymonster}.
\newblock PhD thesis, {U}niversit{\"a}t {B}onn, 1995.
\newblock See: {B}onner {M}athematische {S}chriften 286.

\bibitem[H{\"o}h03]{Hoehn03}
G.~H{\"o}hn.
\newblock Genera of vertex operator algebras and three-dimensional topological
  quantum field theories.
\newblock In {\em Vertex operator algebras in mathematics and physics
  ({T}oronto, {ON}, 2000)}, volume~39 of {\em Fields Inst. Commun.}, pages
  89--107. Amer. Math. Soc., Providence, RI, 2003.

\bibitem[Hua05]{Huang05}
Y.-Z. Huang.
\newblock Vertex operator algebras, the {V}erlinde conjecture, and modular
  tensor categories.
\newblock {\em Proc. Natl. Acad. Sci. USA}, 102(15):5352--5356 (electronic),
  2005.

\bibitem[Hua08]{HuangModularity}
Y.-Z. Huang.
\newblock Rigidity and modularity of vertex tensor categories.
\newblock {\em Commun. Contemp. Math.}, 10(suppl. 1):871--911, 2008.

\bibitem[Lam01]{Lam01}
C.~H. Lam.
\newblock Induced modules for orbifold vertex operator algebras.
\newblock {\em J. Math. Soc. Japan}, 53(3):541--557, 2001.

\bibitem[LS15]{LamShimakura2015}
C.~H. Lam and H.~Shimakura.
\newblock Classification of holomorphic framed vertex operator algebras of
  central charge 24.
\newblock {\em Amer. J. Math.}, 137(1):111--137, 2015.

\bibitem[LY08]{LamYamauchi08}
C.~H. Lam and H.~Yamauchi.
\newblock On the structure of framed vertex operator algebras and their
  pointwise frame stabilizers.
\newblock {\em Comm. Math. Phys.}, 277(1):237--285, 2008.

\bibitem[Mas07]{Mason07}
G.~Mason.
\newblock Vector-valued modular forms and linear differential operators.
\newblock {\em Int. J. Number Theory}, 3(3):377--390, 2007.

\bibitem[Mas08]{MasonVVMF08}
G.~Mason.
\newblock 2-dimensional vector-valued modular forms.
\newblock {\em Ramanujan J.}, 17(3):405--427, 2008.

\bibitem[McR19]{McRae19}
R.~McRae.
\newblock On the tensor structure of modules for compact orbifold vertex
  operator algebras.
\newblock {\em Math. Z.}, 2019.

\bibitem[Miy96a]{Miyamoto96Automorphism}
M.~Miyamoto.
\newblock Automorphism group of {$\mathbb{Z}_2$}-orbifold {VOA}s.
\newblock unpublished paper, 1996.

\bibitem[Miy96b]{Miyamoto96Griess}
M.~Miyamoto.
\newblock Griess algebras and conformal vectors in vertex operator algebras.
\newblock {\em J. Algebra}, 179(2):523--548, 1996.

\bibitem[Miy98]{Miyamoto98}
M.~Miyamoto.
\newblock Representation theory of code vertex operator algebra.
\newblock {\em J. Algebra}, 201(1):115--150, 1998.

\bibitem[MMS88]{MathurMukhiSen88}
S.~D. Mathur, S.~Mukhi, and A.~Sen.
\newblock On the classification of rational conformal field theories.
\newblock {\em Physics Letters B}, 213(3):303--308, 1988.

\bibitem[MNS18]{MasonNagatomoSakai18ax}
G.~Mason, K.~Nagatomo, and Y.~Sakai.
\newblock Vertex operator algebras with two simple modules - the
  {M}athur-{M}ukhi-{S}en theorem revisited.
\newblock {\em arXiv:1803.11281 [math.QA]}, 2018.

\bibitem[RSW09]{RSW}
E.~Rowell, R.~Stong, and Z.~Wang.
\newblock On classification of modular tensor categories.
\newblock {\em Comm. Math. Phys.}, 292(2):343--389, 2009.

\bibitem[Sch93]{Schellekens93}
A.~N. Schellekens.
\newblock Meromorphic {$c=24$} conformal field theories.
\newblock {\em Comm. Math. Phys.}, 153(1):159--185, 1993.

\bibitem[Shi14]{Shimakura14}
H.~Shimakura.
\newblock The automorphism group of the {$\mathbb{Z}_2$}-orbifold of the
  {B}arnes-{W}all lattice vertex operator algebra of central charge 32.
\newblock {\em Math. Proc. Cambridge Philos. Soc.}, 156(2):343--361, 2014.

\bibitem[TW17]{TenerWang17}
J.~E. Tener and Z.~Wang.
\newblock On classification of extremal non-holomorphic conformal field
  theories,.
\newblock {\em J. Phys. A: Math. Theor.}, (50):115204, 2017.

\bibitem[vEMS20]{vanEkerenMoellerScheithauer20}
J.~van Ekeren, S.~M\"{o}ller, and N.~R. Scheithauer.
\newblock Construction and classification of holomorphic vertex operator
  algebras.
\newblock {\em J. Reine Angew. Math.}, 759:61--99, 2020.

\bibitem[Yam04]{Yamauchi04}
H.~Yamauchi.
\newblock Module categories of simple current extensions of vertex operator
  algebras.
\newblock {\em J. Pure Appl. Algebra}, 189(1-3):315--328, 2004.

\bibitem[Zhu96]{Zhu96}
Y.~Zhu.
\newblock Modular invariance of characters of vertex operator algebras.
\newblock {\em J. Amer. Math. Soc.}, 9(1):237--302, 1996.

\end{thebibliography}

\end{document}